\documentclass[12pt,reqno]{amsart} 

\usepackage{mathdots}
\usepackage{amsmath, amssymb} %% AMS symbol definitions
\usepackage{amscd}            %% for Commutative Diagrams
\usepackage{amsxtra}          %% defines \spcheck, \spdot, etc.
\usepackage{upref}   
\usepackage{enumerate}         %% cross-refs to sections are upright
\usepackage{textcomp}

\usepackage[mathscr]{eucal}

\theoremstyle{plain}

\newtheorem{theorem}{Theorem}

\newtheorem{prop}[theorem]{Proposition}

\newtheorem{theo}[theorem]{Theorem}

\newtheorem{lemma}[theorem]{Lemma}
\newtheorem{remark}[theorem]{Remark}

\newtheorem{definition}[theorem]{Definition}
\newtheorem{example}[theorem]{Example}
\theoremstyle{remark}
\newtheorem{case}{Case}

\numberwithin{theorem}{section}
\newcommand{\be}%
  {\protect\setcounter{equation}{\value{subsubsection}}}  
\newcommand{\ee}%
\newcommand{\Z}{\mathbb{Z}}

\newcommand{\F}{\mathbb{F}}

\newcommand{\xn}{x^n - 1}

\begin{document}

%%%%%%%%%%%%%%%%%%%%%%%%%%%%%%%%%%%%%%%%%%%%%%%%%%%%%%%%%%%%%%%%%%%%%
%%%%%%%%%%%%%%%%%%%%%%%%%%%%%%%%%%%%%%%%%%%%%%%%%%%%%%%%%%%%%%%%%%%%%
%% TOPMATTER
%%%%%%%%%%%%%%%%%%%%%%%%%%%%%%%%%%%%%%%%%%%%%%%%%%%%%%%%%%%%%%%%%%%%%

\title {Cyclic codes over the ring $\F_p[u,v]\textfractionsolidus \langle u^k,v^2,uv-vu\rangle$}
%%%%%

%%%%%
\author{Bappaditya Ghosh and Pramod Kumar Kewat} %% This is the correct form!
\address{Department of Applied Mathematics\\ 
         Indian School of Mines\\
         Dhanbad 826 004,  India}
\email{bappaditya.ghosh86@gmail.com, kewat.pk.am@ismdhanbad.ac.in}

\begin{abstract}
Let $p$ be a prime number. In this paper, we discuss the structures of cyclic codes over the ring $ \F_p[u, v]\textfractionsolidus \langle u^k, v^2, uv-vu\rangle$. We find a unique set of generators for these codes. We also study the rank and the Hamming distance of these codes. 
\end{abstract}

\maketitle

%%    LEFT AND RIGHT RUNNING HEADS.  MUST COME *AFTER* \maketitle

\markboth{B. Ghosh and P. K. Kewat}{Cyclic codes over the ring $R_{u^k,v^2,p}$}
%%%%%%%%%%%%%%%%%%%%%%%%%%%%%%%%%%%%%%%%%%%%%%%%%%%%%%%%%%%%%%%%%%%%%
%% BODY OF PAPER BEGINS HERE
%%%%%%%%%%%%%%%%%%%%%%%%%%%%%%%%%%%%%%%%%%%%%%%%%%%%%%%%%%%%%%%%%%%%%
%%%%%%%%%%%%%%%%%%%%%%%%%%%%%%%%%%%%%%%%%%%%%%%%%%%%%%%%%%%%%%%%%%%%%
%% INTRODUCTION
%%%%%%%%%%%%%%%%%%%%%%%%%%%%%%%%%%%%%%%%%%%%%%%%%%%%%%%%%%%%%%%%%%%%%
\section{Introduction}
The study of linear codes over finite rings has accomplished significant important since the realization that some good nonlinear codes can be identified as Gray image of $\Z_4$ linear code (cf. \cite{Nech91, Ham93, Ham94}). Cyclic codes, an important class of linear codes, have also generated great interest in algebraic coding theory. The good progress has been achieved in a series of papers in the direction of determining the structural properties of cyclic codes over the large family of rings, mainly over finite chain rings (cf. \cite{Tah-Oeh04, Tah-Oeh03, Tah-Siap07, Ash-Ham11, Blac03, Bonn-Udaya99, Cal-Slo95, Cal-Slo98, Con-Slo93, Dinh10, Dinh-Lopez04, Dou-Shiro01, Ples-Qian96,  Aks-Pkk13, Lint91}). 

Yildiz and Karadeniz in \cite{Yil-Kar11} have considered the ring $\F_2[u,v]/\langle u^2, v^2, uv-vu \rangle$, which is not a chain ring, and studied cyclic codes of odd length over that. They have found some good binary codes as the Gray images of these cyclic codes. The authors of \cite{Dou-Yil-Kar12} studied the general properties of cyclic codes over the more general ring $\F_2[u_1, u_2, \cdots, u_k]/\langle u_{i}^2, u_{j}^2, u_iu_j-u_ju_i \rangle$ and  characterized the nontrivial one-generator cyclic codes. Sobhani and Molakarimi in \cite{Sob-Mor13} extended these studies to cyclic codes over the ring $\F_{2^m}[u, v]/\langle u^2, v^2, uv-vu \rangle$. The authors of \cite{KGP15} have studied the cyclic codes over the ring $\F_p[u,v]/\langle u^2, v^2,$ $uv-vu \rangle$ and have found some good ternary codes as the Gray images of these cyclic codes.

In this paper, we discuss the structure of cyclic codes of arbitrary length $n$ over the ring $R_{u^k,v^2,p}=\F_p[u,v]/\langle u^k, v^2, uv-vu \rangle$, $k$ a positive integer. We find a unique set of generators for these codes. The idea to find a set of generators is as follows. We view the cyclic code $C$ as an ideal in the ring $R_{u^k,v^2,p,n}= R_{u^k,v^2,p}[x]/\langle\xn \rangle$. We define the projection map from $R_{u^k,v^2,p,n} \longrightarrow R_{u^k,p,n}= R_{u^k,p}[x]/\langle\xn \rangle,$ $R_{u^k,p}=\F_p[u]/\langle u^k\rangle$, and we get an ideal in the ring $R_{u^k,p,n}$, which gives a cyclic code over the ring $R_{u^k,p}$. The structure of cyclic codes over the ring $R_{u^k,p}$  is known from \cite{Aks-Pkk13}. By pullback, we find a set of generators for a cyclic code over the ring $R_{u^k,v^2,p}$. We simplified a set of generators for these cyclic codes when $n$ is relatively prime to $p$. We also provide the characterization of the free cyclic codes over the ring $R_{u^k,v^2,p}$.

We find the rank and  minimal spanning set for a cyclic code $C$ of arbitrary length $n$ over the ring $R_{u^k,v^2,p}$. We find the rank of these cyclic codes by using the division algorithm and direct computations. We first find the minimal spanning sets of kernel and image of the projection map from $C \longrightarrow R_{u^k,p,n}$. Then using the isomorphism theorem we find the minimal spanning set a cyclic code $C$ over the ring $R_{u^k,v^2,p}$. These computations are not straightforward, there are difficulties that need to be overcome. For example, we have several non regular elements as parts of generators, where we can not apply the division algorithm directly. We have used the inductive arguments to find the minimal spanning set in these cases. We also find the Hamming distance of these codes for length $p^l$. 

\section{Preliminaries} \label{pre}
A ring with the unique maximal ideal is called a local ring. Let $R$ be a finite commutative local ring with maximal ideal $M$. Let $\overline{R} = R\textfractionsolidus M$ be the residue field and $\mu : R[x] \rightarrow \overline{R}[x]$ denote the natural ring homomorphism that maps $r \mapsto r + M$ and the variable $x$ to $x$. The degree of the polynomial $f(x) \in R[x]$ as the degree of the polynomial $\mu(f(x))$ in $\overline{R}[x]$, i.e., $deg(f(x)) = deg(\mu(f(x))$ (see, for example, \cite{McDonald74}). A polynomial $f(x) \in R[x]$ is called regular if it is not a zero divisor.
The following conditions are equivalent for a finite commutative local ring $R$.
\begin{prop} {\rm (cf. \cite[Exercise XIII.2(c)]{McDonald74})} \label{regular-poly}
Let $R$ be a finite commutative local ring. Let $f(x) = a_0+a_1x+ \cdots +a_nx^n$ be in $R[x]$, then the following are equivalent.
\begin{enumerate}[{\rm (1)}]
\item $f(x)$ is regular; \label{1}
\item  $\langle a_0, a_1, \cdots , a_n \rangle = R$; \label{4} 
\item  $a_i$ is an unit for some $i$, $0 \leq i \leq n$; \label{3}
\item  $\mu(f(x)) \neq 0$; \label{2}
\end{enumerate}
\end{prop}
%\begin{proof}
% $(1) \Rightarrow (2):$ If possible, let $f(x)$ be a regular polynomial but $\langle a_0, a_1, \cdots ,$ $a_n \rangle$ $\neq R$. Since $\langle a_0, a_1, \cdots , a_n \rangle$ is an ideal and $R$ is a local ring, we have $\langle a_0, a_1, \cdots , a_n \rangle \subseteq M$. Therefore, all $a_i$ are non unit elements. In a finite local ring, if $M=\langle b_1, b_2, \cdots , b_n\rangle$ and the nilpotency index of $b_i$ are $t_i$ respectively, then the element $b_1^{t_1-1}b_2^{t_2-1}\cdots b_n^{t_n-1}$ is the zero divisor of all non unit elements. So the polynomial $g(x)=b_1^{t_1-1}b_2^{t_2-1}\cdots b_n^{t_n-1}$ is the zero divisor of $f(x)$. That is $f(x)g(x)=0$. Hence, a contradiction. Therefore, $\langle a_0, a_1, \cdots , a_n\rangle = R$.\\\\
%$(2) \Rightarrow (3):$ If possible, let all $a_i$ are non unit. Then $a_i\in M$  for all $i$. This implies that $R=M$. This gives a contradiction. Hence, $a_i$ is an unit for some $i$, $0 \leq i \leq n$.\\\\
%$(3)\Rightarrow(4):$ If $a_i$ is an unit, then $a_i \notin M.$ This gives $a_i+M \neq 0$. Hence, $\mu(f(x)) \neq 0$.\\\\
%$(4)\Rightarrow(1):$ Let $\mu(f(x)) \neq 0$. We have $a_i+M \neq 0$ for some $i$. Therefore, $a_i$ is an unit in $R$. So there does not exists non zero element $d \in R$ such that $da_i=0$. By McCoy Theorem (See \cite[Theorem 2]{Mccoy42}), we can say that $f(x)$ is not a divisor zero. That is $f(x)$ is regular.
%\end{proof}

The following version of the division algorithm holds true for polynomials over finite commutative local rings.
\begin{prop}{\rm (cf. \cite[Exercise XIII.6]{McDonald74})} \label{division-alg}
Let $R$ be a finite commutative local ring. Let $f(x)$ and $g(x)$ be non zero polynomials in $R[x]$. If $g(x)$ is regular, then there exist polynomials $q(x)$ and $r(x)$ in $R[x]$ such that $f(x) =g(x)q(x) + r(x)$ and $deg(r(x)) < deg(g(x))$.
\end{prop}

\subsection{{\rm The ring} $R_{u^k,v^2,p}$} $ $

Let $R_{u^k,v^2,p} = (\F_p + u\F_p + \cdots + u^{k-1}\F_p)+v(\F_p + u\F_p + \cdots + u^{k-1}\F_p), u^k=0$, $v^2=0$ and $uv = vu$. Also it can be viewed as $(\F_p+v\F_p) + u(\F_p+v\F_p) + \cdots + u^{k-1}(\F_p+v\F_p), u^k=0$, $v^2=0$ and $uv = vu$. The ring $R_{u^k,v^2,p}$ is a finite commutative local ring with unique maximal ideal $\langle u,v\rangle$. The set $\{ \{0\}, \langle u\rangle,\langle u^2\rangle, \cdots, \langle u^{k-1}\rangle, \langle v\rangle, \langle uv\rangle, \langle u^2v\rangle, \cdots, \langle u^{k-1}v\rangle, \langle u + \alpha v\rangle, \langle u^2 + \alpha v\rangle,$ $ \cdots, \langle u^{k-1} + \alpha v\rangle, \langle u^{k-1},v\rangle,$ $\langle u^{k-2},v\rangle, \cdots, \langle u,v\rangle, \langle 1\rangle\}$ gives list of all ideals of $R_{u^k, v^2, p}$, where $\alpha$ is a non zero element of $\F_p$. Since the maximal ideal $\langle u,v\rangle$ is not principal, the ring $R_{u^k, v^2, p}$ is not a chain ring. 

Let $g(x)$ be a non zero polynomial in $\F_p[x]$. By Proposition \ref{regular-poly}, it is easy to see that the polynomial $g(x)+up_1(x)+\cdots+u^{k-1}p_{k-1}(x)+v(p_k(x)+up_{k+1}(x)+\cdots+u^{k-1}p_{2k-1}(x)) \in R_{u^k,v^2,p}[x]$ is regular. Note that $\text{deg}(g(x)+up_1(x)+\cdots+u^{k-1}p_{k-1}(x)+v(p_k(x)+up_{k+1}(x)+\cdots+u^{k-1}p_{2k-1}(x))) = \text{deg}(g(x))$.

\subsection{{\rm The Gray map}} $ $

The Gray map $\varphi_L:R_{u^k,v^2,p}\rightarrow \F_p^{2k}$ is defined as follows
\[\varphi_L\left(a+vb\right)\rightarrow \left(\varphi\left(a+b\right),\varphi\left(b\right)\right), ~\forall ~ a,b \in R_{u^k,p},\]
where
\[\varphi\left(a_1+ua_2+\cdots+u^{k-1}a_k\right)=\left(\sum\limits_{i=1}^ka_i, \sum\limits_{i=2}^ka_i, \sum\limits_{i=2}^{k-1}a_i, \cdots, \sum\limits_{i=\frac{k+1}{2}}^{\frac{k+3}{2}}a_i, \sum\limits_{i=\frac{k+1}{2}}^{\frac{k+1}{2}}a_i\right),\]
when $k$ is odd and
\[\varphi\left(a_1+ua_2+\cdots+u^{k-1}a_k\right)=\left(\sum\limits_{i=1}^ka_i, \sum\limits_{i=2}^ka_i, \sum\limits_{i=2}^{k-1}a_i, \cdots, \sum\limits_{i=\frac{k}{2}}^{\frac{k}{2}+1}a_i, \sum\limits_{i=\frac{k}{2}+1}^{\frac{k}{2}+1}a_i\right),\]
when $k$ is even.

Let $w_L$ and $w_H$ denote the Lee weight and Hamming weight respectively. We define the Lee weight as follows
\[w_L\left(a+vb\right)=w_H\left(\varphi_L\left(a+vb\right)\right), ~\forall ~ a,b \in R_{u^k,p},\]
The Gray map naturally extend to $R_{u^k,v^2,p}^n$ as distance preserving isometry
\[\varphi_L:\left(R_{u^k,v^2,p}^n, ~ \text{Lee weight}\right) \rightarrow \left(\F_p^{2kn}, ~ \text{Hamming weight}\right)\]
as follows
\[\varphi_L\left(a_1, a_2, \cdots, a_n\right)\rightarrow \left(\varphi_L\left(a_1\right),\varphi_L\left(a_2\right), \cdots, \varphi_L\left(a_n\right)\right), ~\forall ~ a_i \in R_{u^k,v^2,p}.\]
By linearity of the map $\varphi_L$ we obtain the following theorem.
\begin{theo}
If $C$ is a linear code over the ring $R_{u^k, v^2, p}$ of length $n$, size $p^l$ and
minimum Lee weight $d$, then $\varphi_L(C)$ is a $p$-ary linear code with parameters $[2kn, l, d]$.

\end{theo}

\section{The structures of cyclic codes over the ring $R_{u^k, v^2, p}$} \label{generator}
Let $p$ be a prime number and $n$ be a positive integer. Let $R_{u^k,v^2,p} = \F_p[u, v]\textfractionsolidus \langle u^k,v^2,uv-vu \rangle$. We can write $R_{u^k,v^2,p}$ as $R_{u^k,v^2,p}= R_{u^k,p} + vR_{u^k,p}, v^2=0$, where $R_{u^k,p}=\F_p + u\F_p + \cdots + u^{k-1}\F_p$ and $u^k=0$. Also, it can be written as $R_{u^k,v^2,p}=R_{v^2,p} + uR_{v^2,p} + \cdots + u^{k-1}R_{v^2,p}$ and $u^k=0$, where $R_{v^2,p}=\F_p+v\F_p$ and $v^2=0$. Let $R_{u^k,v^2,p,n}=R_{u^k,v^2,p}[x]\textfractionsolidus \langle x^n-1\rangle$. Let $C$ be a cyclic code of length $n$ over $R_{u^k,v^2,p}$. We also consider $C$ as an ideal in the ring $R_{u^k,v^2,p,n}$. We define the map $\psi : R_{u^k,v^2,p} \rightarrow R_{u^k,p}$ by $\psi(\alpha + v \beta) = \alpha$, where $\alpha, \beta \in R_{u^k,p}$. Clearly the map $\psi$ is a surjective ring homomorphism. Let $R_{u^k,p,n}=R_{u^k,p}[x]\textfractionsolidus \langle x^n-1\rangle$. We extend this homomorphism to a homomorphism $\phi$ from $C$ to the ring $R_{u^k,p,n}$ 
defined by
\begin{equation} \label{surj-hom}
\phi\left(c_0+c_1x+\cdots+c_{n-1}x^{n-1}\right)=\psi\left(c_0\right)+\psi\left(c_1\right)x+\cdots+\psi\left(c_{n-1}\right)x^{n-1},
\end{equation}
where $c_i \in R_{u^k,v^2,p}$. Let $J=\{r(x)\in R_{u^k,p,n}[x] ~ | ~ vr(x) \in \text{ker}\phi\}$. We see that $J$ is an ideal of $R_{u^k,p,n}$. Hence we can consider $J$ as a cyclic code over the ring $R_{u^k,p}$. We know from Theorem 3.3 of \cite{Aks-Pkk13} that any ideal of $R_{u^k,p,n}$ is of the form $\langle g(x) + u p_1(x) + u^2 p_{2}(x) + \cdots + u^{k-1} p_{k-1}(x), u a_1(x) + u^2 q_1(x) + \cdots + u^{k-1} q_{k-2}(x), u^2 a_2(x) + u^3 l_{1}(x) + \cdots + u^{k-1}l_{k-3}(x), \cdots, u^{k-2} a_{k-2}(x) + u^{k-1} t_1(x), u^{k-1} a_{k-1}(x)\rangle$ with $a_{k-1}(x)|a_{k-2}(x)| \cdots | a_2(x)|a_1(x)|g(x)|(\xn)$ mod $p$ ; $a_1(x)|p_1(x)\frac{\xn}{g(x)}$, $\cdots$, $a_{k-1}(x)|t_1(x)\frac{\xn}{a_{k-2}(x)}$; $\cdots$ ; $a_{k-1}(x)|p_{k-1}(x)\frac{\xn}{g(x)}\cdots\frac{\xn}{a_{k-2}(x)}$; ${\rm deg}(p_1(x))< {\rm deg}(a_1(x))$; ${\rm deg}(p_2(x))$, ${\rm deg}(q_1(x))<{\rm deg}(a_2(x))$; $\cdots$ ; ${\rm deg}(p_{k-1}(x))$, ${\rm deg}(q_{k-2}(x))$, $\cdots$, ${\rm deg}(t_1(x))<{\rm deg}(a_{k-1}(x))$. Now we 
assume 
that $B_1 = g(x) + u p_1(x) + u^2 p_{2}(x) + \cdots + u^{k-1} p_{k-1}(x), B_2 = u a_1(x) + u^2 q_1(x) + \cdots + u^{k-1} q_{k-2}(x), B_3 = u^2 a_2(x) + u^3 l_{1}(x) + \cdots + u^{k-1}l_{k-3}(x)$, $\cdots$, $B_{k-1} = u^{k-2} a_{k-2}(x) + u^{k-1} t_1(x)$, $B_k = u^{k-1} a_{k-1}(x)$. So $J = \langle B_1, B_2, \cdots, B_k\rangle$. Therefore, we can write $\text{ker}\phi=\langle vB_1, vB_2, \cdots, vB_k\rangle$. Since $\phi$ is a surjective homomorphism, the image $\text{Im}\phi$ is an ideal of $R_{u^k,p,n}$. Hence, $\text{Im}\phi$ is a cyclic code over the ring $R_{u^k,p}$. Again we can write $\text{Im}\phi$ as above. That is, $\text{Im}\phi=\langle B_1', B_2', \cdots, B_k'\rangle$. Therefore the code $C$ over the ring $R_{u^k,v^2,p}$ can be written as $C=\langle A_1, A_2, \cdots, A_{2k}\rangle$, where, $A_i$'s are defined as follows.
\label{A_i's}
\begin{align*}
A_1 = &~ g_1(x)+ug_{11}(x)+u^2g_{12}(x)+\cdots+u^{k-1}g_{1(k-1)}(x)\\ & +v(g_{1k}(x)+ug_{1(k+1)}(x)+u^2g_{1(k+2)}(x) + \cdots +u^{k-1}g_{1(2k-1)}(x)),\\
A_2 = &~ ug_2(x)+u^2g_{22}(x)+u^3g_{23}(x)+\cdots+u^{k-1}g_{2(k-1)}(x)\\ & +v(g_{2k}(x)+ug_{2(k+1)}(x)+u^2g_{2(k+2)}(x) + \cdots +u^{k-1}g_{2(2k-1)}(x)),\\
\vdots\\
A_i = &~ u^{i-1}g_i(x)+u^ig_{ii}(x)+u^{i+1}g_{i(i+1)}(x)+ \cdots +u^{k-1}g_{i(k-1)}(x)\\ & +v(g_{ik}(x)+ug_{i(k+1)}(x)+u^2g_{i(k+2)}(x)+ \cdots +u^{k-1}g_{i(2k-1)}(x)),\\
\vdots\\
A_{k-1} = &~ u^{k-2}g_{k-1}(x)+u^{k-1}g_{(k-1)(k-1)}(x)+v(g_{(k-1)k}(x)+ug_{(k-1)(k+1)}(x)\\ & +u^2g_{(k-1)(k+2)}(x)+  \cdots+u^{k-1}g_{(k-1)(2k-1)}(x)),\\
A_k = &~ u^{k-1}g_k(x)+v(g_{kk}(x)+ug_{k(k+1)}(x)+u^2g_{k(k+2)}(x)+\cdots+\\ & u^{k-1}g_{k(2k-1)}(x)),\\
A_{k+1} = &~ v(g_{k+1}(x)+ug_{(k+1)(k+1)}(x)+u^2g_{(k+1)(k+2)}(x)+\cdots+\\ & u^{k-1}g_{(k+1)(2k-1)}(x)),\\
A_{k+2} = &~ v(ug_{k+2}(x)+u^2g_{(k+2)(k+2)}(x)+u^3g_{(k+2)(k+3)}(x)+\cdots+\\ & u^{k-1}g_{(k+2)(2k-1)}(x)),\\
\vdots\\
A_{k+i} = &~ v(u^{i-1}g_{k+i}(x)+u^ig_{(k+i)(k+i)}(x)+u^{i+1}g_{(k+i)(k+i+1)}(x)+ \cdots +\\ & u^{k-1}g_{(k+i)(2k-1)}(x)),\\
\vdots\\
A_{2k-1} = &~ v(u^{k-2}g_{2k-1}(x)+u^{k-1}g_{(2k-1)(2k-1)}(x)),\\
A_{2k} = &~ vu^{k-1}g_{2k}(x).
\end{align*}
Throughout this paper we use $A_1, A_2, \cdots, A_{2k}$ for above polynomials.\\

For an ideal $C$ of the ring $R_{u^k,v^2,p,n}=R_{u^k,v^2,p}[x]\textfractionsolidus \langle x^n-1\rangle$, we define the residue and the torsion of the ideal $C$ as (see \cite{Dou-Yil-Kar12})
\begin{align*}
&\text{Res}(C)=\{a\in R_{u^k,p,n}|~\exists ~b\in R_{u^k,p,n}: a+vb\in C\} ~\text{and}\\
&\text{Tor}(C)=\{a\in R_{u^k,p,n}|~va\in C\}
\end{align*}
It is easy to see that when $C$ is an ideal of the ring $R_{u^k,v^2,p,n}$, the $\text{Res}(C)$ and $\text{Tor}(C)$ both are ideals of $R_{u^k,p,n}$. And also it is easy to show that $\text{Res}(C)=\text{Im}\phi$ and $\text{Tor}(C)=J$.\\

Again, for an ideal $C'$ of the ring $R_{u^i,p,n}=R_{u^i,p}[x]\textfractionsolidus \langle x^n-1\rangle$, for $2 \leq i \leq k$ we define residue and torsion of the ideal $C'$ as
\begin{align*}
&\text{Res}(C')=\{a\in R_{u^{i-1},p,n}|~\exists ~b\in R_{u^{i-1},p,n}: a+u^{i-1}b\in C'\} ~\text{and}\\
&\text{Tor}(C')=\{a\in R_{u,p,n}|~u^{i-1}a\in C'\}
\end{align*}
Here $\text{Res}(C')$ and $\text{Tor}(C')$ are ideals of  the ring $R_{u^{i-1},p,n}$ and $R_{u,p,n}$ respectively.\\
Note that $R_{u,p}=\frac{\F_p[u]}{\langle u \rangle} \simeq \F_p$, therefore, $R_{u,p,n}=\frac{\F_p[x]}{\langle x^n-1\rangle}$.\\
Now we define the ideals $C_1, ~C_2, \cdots, C_{2k}$ associated to $C$ as follows.
\label{C_i's}
\begin{align*}
&C_1= \underbrace{\text{Res}\cdots \text{Res}}_{k~times}(C)=C~\text{mod} ~ \langle u,v\rangle=\langle g_1(x)\rangle\\
&C_2=\text{Tor}\underbrace{\text{Res}\cdots \text{Res}}_{k-1~times}(C)=\{f(x)\in \F_p[x]~|~uf(x)\in C~\text{mod} ~ \langle u^2,v\rangle\}=\langle g_2(x)\rangle\\ 
&\vdots\\
&C_i=\text{Tor}\underbrace{\text{Res}\cdots \text{Res}}_{k-i+1~times}(C)=\{f(x)\in \F_p[x]~|~u^{i-1}f(x)\in C~\text{mod} ~ \langle u^i,v\rangle\}\\
&~\hspace{.7cm} =\langle g_i(x)\rangle\\
&\vdots\\
&C_{k-1}=\text{Tor}\text{Res}\text{Res}(C)=\{f(x)\in \F_p[x]~|~u^{k-2}f(x)\in C~\text{mod} ~ \langle u^{k-1},v\rangle\} \\
&~ \hspace{.7cm}=\langle g_{k-1}(x)\rangle\\
&C_k=\text{Tor}\text{Res}(C)=\{f(x)\in \F_p[x]~|~u^{k-1}f(x)\in C~\text{mod} ~ \langle v\rangle\}=\langle g_k(x)\rangle\\
&C_{k+1}=\underbrace{\text{Res}\cdots \text{Res}}_{k-1~times}\text{Tor}(C)=\{f(x)\in \F_p[x]~|~vf(x)\in C~\text{mod} ~ \langle uv\rangle\}\\
&~\hspace{.7cm} =\langle g_{k+1}(x)\rangle\\
&C_{k+2}=\text{Tor}\underbrace{\text{Res}\cdots \text{Res}}_{k-2~times}\text{Tor}(C)=\{f(x)\in \F_p[x]~|~uvf(x)\in C~\text{mod} ~ \langle u^2v\rangle\} \\ 
&~\hspace{.7cm} =\langle g_{k+2}(x)\rangle\\ 
&\vdots\\
&C_{k+i}=\text{Tor}\underbrace{\text{Res}\cdots \text{Res}}_{k-i~times}\text{Tor}(C)=\{f(x)\in \F_p[x]~|~u^{i-1}vf(x)\in C~\text{mod} ~ \langle u^iv\rangle\} \\ 
&~\hspace{.7cm} =\langle g_{k+i}(x)\rangle\\
&\vdots\\
&C_{2k-1}=\text{Tor}\text{Res}\text{Tor}(C)=\{f(x)\in \F_p[x]~|~u^{k-2}vf(x)\in C~\text{mod} ~ \langle u^{k-1}v\rangle\}\\ &~\hspace{.7cm}=\langle g_{2k-1}(x)\rangle\\
&C_{2k}=\text{Tor}\text{Tor}(C)=\{f(x)\in \F_p[x]~|~u^{k-1}vf(x)\in C\}=\langle g_{2k}(x)\rangle
\end{align*}
Here all $C_i$'s are ideals of $\frac{\F_p[x]}{\langle x^n-1\rangle}$. Throughout this paper we use $C_1, ~C_2, \cdots, C_{2k}$ for above ideals.

\begin{theo}\label{unique}
Any ideal $C$ of the ring $R_{u^k,v^2,p,n}$ is uniquely generated by the polynomials $A_1, A_2, \cdots, A_{2k}$ with $g_{ij}(x)$ are zero polynomials or $\text{deg}(g_{ij}(x))<\text{deg}(g_{j+1}(x))$ for $1\leq i\leq (2k-1)$, $i\leq j\leq (2k-1)$ where $A_i$ and $g_{ij}(x)$'s are defined as on page \pageref{A_i's} $($see page \pageref{A_i's}$)$.
\end{theo}
\begin{proof}
We prove the degree result for $g_{1j}(x)$, where $1\leq j\leq 2k-1$. Others are similar. Let $A_1\neq 0$ and $\text{deg}(g_{11}(x))\geq \text{deg}(g_2(x))$. Then by division algorithm, we have $g_{11}(x)=q_1(x)g_2(x)+r_1(x)$, where $\text{deg}(r_1(x))<\text{deg}(g_2(x))$ or $r_1(x)=0$. Now $A_1-q_1(x)A_2=g_1(x)+ur_1(x)+u^2(g_{12}(x)-q_1(x)g_{22}(x))+u^3(g_{13}(x)-q_1(x)g_{23}(x))+\cdots+u^{k-1}(g_{1(k-1)}(x)-q_1(x)g_{2(k-1)}(x))+v((g_{1k}(x)-q_1(x)g_{2k}(x))+u(g_{1(k+1)}(x)-q_1(x)g_{2(k+1)}(x))+u^2(g_{1(k+2)}(x)-q_1(x)g_{2(k+2)}(x))+\cdots+u^{k-1}(g_{1(2k-1)}(x)-q_1(x)g_{2(2k-1)}(x)))$. If $\text{deg}(g_{12}(x)-q_1(x)g_{22}(x))\geq\linebreak \text{deg}(g_3(x))$, then by division algorithm, $g_{12}(x)-q_1(x)g_{22}(x)=q_2(x)g_3(x)+r_2(x)$, where $\text{deg}(r_2(x))<\text{deg}(g_3(x))$ or $r_2(x)=0$. We have $A_1-q_1(x)A_2-q_2(x)A_3=g_1(x)+ur_1(x)+u^2r_2(x)+u^3(g_{13}(x)-q_1(x)g_{23}(x)-q_2(x)g_{33}(x))+\cdots+u^{k-1}(g_{1(k-1)}(x)-q_1(x)g_{2(k-1)}(x)-q_2(x)g_{3(k-1)}(x))+v((g_{1k}(x)-q_1(x)g_{2k}(x)-\
\linebreak q_2(x)g_{3k}(x))+u(
g_{
1(k+1)}(x)-q_1(x)g_{2(k+1)}(x)-q_2(x)g_{3(k+1)}(x))+u^2(g_{1(k+2)}(x)-q_1(x)g_{2(k+2)}(x)-q_2(x)g_{3(k+2)}(x))+\cdots+u^{k-1}(g_{1(2k-1)}(x)-q_1(x)g_{2(2k-1)}(x)-q_2(x)g_{3(2k-1)}(x)))$. Proceeding in this way, after $2k-2$ times, we get $A_1-q_1(x)A_2-q_2(x)A_3-\cdots -q_{2k-2}(x)A_{2k-1}=g_1(x)+ur_1(x)+u^2r_2(x)+\cdots+u^{k-1}r_{k-1}(x)+v(r_k(x)+ur_{k+1}(x)+\cdots+u^{k-2}r_{2k-2}(x)+u^{k-1}(g_{1(2k-1)}-\linebreak q_1(x)g_{2(2k-1)}(x)-q_2(x)g_{3(2k-1)}-\cdots -q_{2k-2}(x)g_{(2k-1)(2k-1)}(x)))$. If $\text{deg}(g_{1(2k-1)}-q_1(x)g_{2(2k-1)}(x)-q_2(x)g_{3(2k-1)}-\cdots -q_{2k-2}(x)g_{(2k-1)(2k-1)}(x))\geq \text{deg}(g_{2k}(x))$, then again by division algorithm, $g_{1(2k-1)}-q_1(x)g_{2(2k-1)}(x)-q_2(x)g_{3(2k-1)}-\cdots -q_{2k-2}(x)g_{(2k-1)(2k-1)}(x)=q_{2k-1}(x)g_{2k}(x)+r_{2k-1}(x)$, where $\text{deg}(r_{2k-1}(x))<\text{deg}(g_{2k}(x))$or $r_{2k-1}(x)=0$. Now $A_1-q_1(x)A_2-q_2(x)A_3-\cdots -q_{2k-2}(x)A_{2k-1}-q_{2k-1}(x)A_{2k}=g_1(x)+ur_1(x)+u^2r_2(x)+\cdots+u^{k-1}r_{k-1}(x)+v(r_k(x)+ur_{k+1}(x)+\cdots+u^{
k-2}r_{2k-2}(x)+u^{k-1}r_{2k-1}(x))$. This polynomial satisfies the required properties of the theorem and also the polynomial $A_1$ can be replaced by this polynomial. Now we have to prove that the polynomials $A_i$'s are unique. Here again, we prove the uniqueness only for polynomial $A_1$. Others are similar. If possible, let $A_1=g_1(x)+ug_{11}(x)+u^2g_{12}(x)+\cdots +u^{k-1}g_{1(k-1)}(x)+v(g_{1k}(x)+ug_{1(k+1)}(x)+u^2g_{1(k+2)}(x)+\cdots +u^{k-1}g_{1(2k-1)}(x))$ and $B_1=g_1(x)+ug'_{11}(x)+u^2g'_{12}(x)+\cdots +u^{k-1}g'_{1(k-1)}+v(g'_{1k}(x)+ug'_{1(k+1)}(x)+u^2g'_{1(k+2)}(x)+\cdots +u^{k-1}g'_{1(2k-1)}(x))$ are two polynomials with same properties in $C$. Hence, $A_1-B_1=u(g_{11}(x)-g'_{11}(x))+u^2(g_{12}(x)-g'_{12}(x))+\cdots +u^{k-1}(g_{1(k-1)}(x)-g'_{1(k-1)}(x))+v((g_{1k}(x)-g'_{1k}(x))+u(g_{1(k+1)}(x)-g'_{1(k+1)}(x))+u^2(g_{1(k+2)}(x)-g'_{1(k+2)}(x))+\cdots +u^{k-1}(g_{1(2k-1)}(x)-g'_{1(2k-1)}(x)))$. We have $A_1-B_1\in C$ which implies that $g_{11}(x)-g'_{11}(x)\in C_2=\langle g_2(x)\rangle$. 
Previously we have proved that the degree of both $g_{11}(x)$ and $g'_{11}(x)$ is less than degree of $g_2(x)$. Hence, $\text{deg}(g_{11}(x)-g'_{11}(x))< \text{deg}(g_2(x))$. But $g_2(x)$ is the minimum degree polynomial in $C_2$, which implies that $g_{11}(x)-g'_{11}(x)=0$. This gives $g_{11}(x)=g'_{11}(x)$. Now $A_1-B_1=u^2(g_{12}(x)-g'_{12}(x))+\cdots +u^{k-1}(g_{1(k-1)}(x)-g'_{1(k-1)}(x))+v((g_{1k}(x)-g'_{1k}(x))+u(g_{1(k+1)}(x)-g'_{1(k+1)}(x))+u^2(g_{1(k+2)}(x)-g'_{1(k+2)}(x))+\cdots +u^{k-1}(g_{1(2k-1)}(x)-g'_{1(2k-1)}(x)))$. We have $A_1-B_1\in C$ which implies that $g_{12}(x)-g'_{12}(x)\in C_3=\langle g_3(x)\rangle$. Again, we have already proved that the degrees of $g_{12}(x)$ and $g'_{12}(x)$ is less than degree of $g_3(x)$. Hence, $\text{deg}(g_{12}(x)-g'_{12}(x))< \text{deg}(g_3(x))$, which implies that $g_{12}(x)-g'_{12}(x)=0$. This gives $g_{12}(x)=g'_{12}(x)$. Similarly we can show $g_{1i}(x)=g'_{1i}(x)$ for all $1\leq i\leq (2k-1)$. Hence, $A_1-B_1=0$. Thus $A_1=B_1$. 
Thus $A_1$ is unique.
\end{proof}

\begin{theo}\label{properties}
Let $C=\langle A_1,A_2, \cdots, A_{2k} \rangle$ be an ideal of the ring $R_{u^k,v^2,p,n}$. Then the following relations hold in the ring $\frac{\F_p[x]}{\langle x^n-1\rangle}$.
\begin{align}
&g_{2k}(x)|g_{2k-1}(x)| \cdots |g_{k+2}(x)|g_{k+1}(x) ~ {\rm{and}} ~ g_k(x)|g_{k-1}(x)| \cdots |g_2(x)|g_1(x)|(x^n-1),\\
&g_{k+i}(x)|g_i(x), ~ {\rm{for}} ~ 1\leq i \leq k,\\
&g_{i+1}(x)|\frac{x^n-1}{g_i(x)}g_{ii}(x), ~ {\rm{for}} ~ 1\leq i \leq 2k-1,\\
&{\rm{For ~ a ~ fix}} ~ j, ~ {\rm{where}}~ 1 \leq j \leq 2k-1,\notag \\
&g_{i+j}(x)|\frac{x^n-1}{g_i(x)}\frac{x^n-1}{g_{i+1}(x)} \cdots \frac{x^n-1}{g_{i+j-1}(x)}g_{i(i+j-1)}(x), ~ {\rm{for}} ~ 1 \leq i \leq 2k-j\\
&g_{k+i}(x)|g_{(k-(i-2))k}(x), ~ {\rm{for}} ~ 2\leq i \leq k,\\
&g_{k+i+1}(x)|r_{ii}(x),~ {\rm{for}} ~ 1\leq i \leq k-1, ~ {\rm{where}}\notag \\
&r_{ii}(x)=g_{ii}(x)-\frac{g_i(x)}{g_{k+i}(x)}g_{(k+i)(k+i)}(x),\\
&g_{k+i+j+1}(x)|r_{i(i+j)}(x), ~ {\rm{for}} ~ 1 \leq i \leq k-2 ~ {\rm{and}} ~ 1 \leq j \leq k-i-1, ~ {\rm{where}}\notag \\
&r_{i(i+j)}(x)=g_{i(i+j)}(x)-\frac{g_i(x)}{g_{k+i}(x)}g_{(k+i)(k+i+j)}(x)-\sum_{l=1}^j\frac{r_{i(i+l-1)}(x)}{g_{k+i+l}(x)}g_{(k+i+l)(k+i+j)}(x)\\
&{\rm{and}}\notag \\
&g_{i+j+1}(x)|\frac{x^n-1}{g_i(x)}s_{i(i+j)}(x), ~ {\rm{for}} ~ 1 \leq i \leq 2k-2 ~ {\rm{and}} ~ 1 \leq j \leq 2k-i-1,\notag \\
&{\rm{where}} ~ s_{ii}(x)=g_{ii}(x) ~ {\rm{and}} ~ s_{i(i+j)}(x)=g_{i(i+j)}(x)-\sum_{l=1}^j\frac{s_{i(i+l-1)}(x)}{g_{i+l}(x)}g_{(i+l)(i+j)}(x).
\end{align}
\end{theo}
\begin{proof}
\begin{enumerate}
\setcounter{enumi}{1}
\item For $1\leq i \leq k-1$, we have $uA_i\in C$. Therefore, $g_i(x)\in C_{i+1}=\langle g_{i+1}(x)\rangle$. This gives $g_{i+1}(x)|g_i(x)$. Again, for $1\leq i \leq k-1$, we have $uA_{k+i}\in C$. Therefore, $g_{k+i}(x)\in C_{k+i+1}=\langle g_{k+i+1}(x)\rangle$. Thus, $g_{k+i+1}(x)|g_{k+i}(x)$.\\
\item For $1\leq i \leq k$, we have $vA_i\in C$. This gives $g_i(x)\in C_{k+i}=\langle g_{k+i}(x)\rangle$. Thus, $g_{k+i}(x)|g_i(x)$.\\
\item For $1\leq i \leq 2k-1$, we have $\frac{x^n-1}{g_i(x)}A_i\in C$. Therefore, $\frac{x^n-1}{g_i(x)}g_{ii}(x)\in C_{i+1}=\langle g_{i+1}(x)\rangle$. Hence, $g_{i+1}(x)|\frac{x^n-1}{g_i(x)}g_{ii}(x)$.\\
\item For $j=1$, Condition 5 is reduced to Condition 4. For $j=2$ and for $1 \leq i \leq k$, we have $\frac{x^n-1}{g_i(x)}\frac{x^n-1}{g_{i+1}(x)}A_i=u^i\frac{x^n-1}{g_i(x)}\frac{x^n-1}{g_{i+1}(x)}g_{ii}(x)+ u^{i+1}\frac{x^n-1}{g_i(x)}\frac{x^n-1}{g_{i+1}(x)}\linebreak g_{i(i+1)}(x)+ \cdots +u^{k-1}\frac{x^n-1}{g_i(x)}\frac{x^n-1}{g_{i+1}(x)}g_{i(k-1)}(x)+ v(\frac{x^n-1}{g_i(x)}\frac{x^n-1}{g_{i+1}(x)}g_{ik}(x)+u\frac{x^n-1}{g_i(x)}\linebreak \frac{x^n-1}{g_{i+1}(x)}g_{i(k+1)}(x)+ \cdots + u^{k-1}\frac{x^n-1}{g_i(x)}\frac{x^n-1}{g_{i+1}(x)}g_{i(2k-1)}(x))$. From Condition 4, we have $g_{i+1}(x)|\frac{x^n-1}{g_i(x)}g_{ii}(x)$. Therefore the coefficient of $u^i$ is\linebreak $\frac{x^n-1}{g_i(x)}\frac{x^n-1}{g_{i+1}(x)}g_{ii}(x)=\frac{x^n-1}{g_{i+1}(x)}(\frac{x^n-1}{g_i(x)}g_{ii}(x))=0$. Thus, $\frac{x^n-1}{g_i(x)}\frac{x^n-1}{g_{i+1}(x)}A_i=u^{i+1}\frac{x^n-1}{g_i(x)}\frac{x^n-1}{g_{i+1}(x)}g_{i(i+1)}(x)+ \cdots +u^{k-1}\frac{x^n-1}{g_i(x)}\frac{x^n-1}{g_{i+1}(x)}g_{i(k-1)}(x)+v(\frac{x^n-1}{g_i(x)}\frac{x^n-
1}{g_{i+1}(x)}\linebreak g_{ik}(x)+u\frac{x^n-1}{g_i(x)}\frac{x^n-1}{g_{i+1}(x)}g_{i(k+1)}(x)+ \cdots +u^{k-1}\frac{x^n-1}{g_i(x)}\frac{x^n-1}{g_{i+1}(x)}g_{i(2k-1)}(x))$. This gives $\frac{x^n-1}{g_i(x)}\frac{x^n-1}{g_{i+1}(x)}g_{i(i+1)}(x) \in C_{i+2}=\langle g_{i+2}(x)\rangle$. Hence, we have\linebreak $g_{i+2}(x)|\frac{x^n-1}{g_i(x)}\frac{x^n-1}{g_{i+1}(x)}g_{i(i+1)}(x)$. Now for $k+1 \leq i \leq 2k-2$. Let $i=k+l$ for $1 \leq l \leq k-2$. We have $\frac{x^n-1}{g_{k+l}(x)}\frac{x^n-1}{g_{k+l+1}(x)}A_{k+l}=v(u^l\frac{x^n-1}{g_{k+l}(x)}\frac{x^n-1}{g_{k+l+1}(x)}\linebreak g_{(k+l)(k+l)}(x)+u^{l+1}\frac{x^n-1}{g_{k+l}(x)}\frac{x^n-1}{g_{k+l+1}(x)}g_{(k+l)(k+l+1)}(x)+ \cdots +u^{k-1}\frac{x^n-1}{g_{k+l}(x)}\frac{x^n-1}{g_{k+l+1}(x)}\linebreak g_{(k+l)(2k-1)}(x))$. From Condition 4, $g_{k+l+1}(x)|\frac{x^n-1}{g_{k+l}(x)}g_{(k+l)(k+l)}(x)$. Therefore, the coefficient of $u^lv$ is $\frac{x^n-1}{g_{k+l}(x)}\frac{x^n-1}{g_{k+l+1}(x)}g_{(k+l)(k+l)}(x)=\frac{x^n-1}{g_{k+l+1}(x)}\frac{x^n-1}{g_{k+l}(x)}\linebreak g_{
(k+l)(k+l)}(x)=0$. This gives $\frac{x^n-1}{g_{k+l}(x)}\frac{x^n-1}{g_{k+l+1}(x)}g_{(
k+l)(k+l+1)}(x) \in C_{k+l+2}=\langle g_{k+l+2}(x)\rangle$. Hence, we have $g_{k+l+2}(x)|\frac{x^n-1}{g_{k+l}(x)}\frac{x^n-1}{g_{k+l+1}(x)}g_{(k+l)(k+l+1)}(x)$. Since $i=k+l$, thus $g_{i+2}(x)|\frac{x^n-1}{g_i(x)}\frac{x^n-1}{g_{i+1}(x)}g_{i(i+1)}(x)$, for $k+1 \leq i \leq 2k-2$. This proves the condition for $j=2$. Similarly for others value of $j$ we can prove the Conditions 5.\\
\item We have $A_i=u^{i-1}g_i(x)+u^ig_{ii}(x)+u^{i+1}g_{i(i+1)}(x)+ \cdots +u^{k-1}g_{i(k-1)}(x)+v(g_{ik}(x)+ug_{i(k+1)}(x)+u^2g_{i(k+2)}(x)+ \cdots +u^{k-1}g_{i(2k-1)}(x))$, for $1\leq i \leq k$. Therefore, $A_{k-(i-2)}=u^{k-(i-1)}g_{k-(i-2)}(x)$ $+$ $u^{k-(i-2)}g_{(k-(i-2))(k-(i-2))}(x)$ $+$ $u^{k-(i-3)}g_{(k-(i-2))(k-(i-3))}(x)$ $+$ $\cdots$ $+$ $u^{k-1}g_{(k-(i-2))(k-1)}(x)+\linebreak v(g_{(k-(i-2))k}(x)$ $+$ $ug_{(k-(i-2))(k+1)}(x)$ $+$ $u^2g_{(k-(i-2))(k+2)}(x)$ $+$ $\cdots$ $+$\linebreak $u^{k-1}g_{(k-(i-2))(2k-1)}(x))$, for $2\leq i \leq k$. Thus, $u^{i-1}A_{k-(i-2)}=\linebreak v(u^{i-1}g_{(k-(i-2))k}(x)+u^ig_{(k-(i-2))(k+1)}(x) +u^{i+1}g_{(k-(i-2))(k+2)}(x)+ \cdots +u^{i+(k-i-2)}g_{(k-(i-2))(2k-i-1)}(x)+u^{i+(k-i-1)}g_{(k-(i-2))(2k-i)}(x))\in C$. This implies that $g_{(k-(i-2))k}(x)\in C_{k+i}=\langle g_{k+i}(x)\rangle$. Hence, the condition $g_{k+i}(x)|g_{(k-(i-2))k}(x)$ for $2\leq i \leq k$ is proved.\\
\item For $1\leq i \leq k-1$, $A_i=u^{i-1}g_i(x)+u^ig_{ii}(x)+u^{i+1}g_{i(i+1)}(x)+ \cdots +u^{k-1}g_{i(k-1)}(x)+v(g_{ik}(x)+ug_{i(k+1)}(x)+u^2g_{i(k+2)}(x)+ \cdots +u^{k-1}g_{i(2k-1)}(x))$ and $A_{k+i}=v(u^{i-1}g_{k+i}(x)+u^ig_{(k+i)(k+i)}(x)+u^{i+1}g_{(k+i)(k+i+1)}(x)+ \cdots +u^{k-1}g_{(k+i)(2k-1)}(x))$. Now $vA_i-\frac{g_i(x)}{g_{k+i}(x)}A_{k+i}=v(u^i(g_{ii}(x)-\frac{g_i(x)}{g_{k+i}(x)}\linebreak g_{(k+i)(k+i)}(x))+u^{i+1}(g_{i(i+1)}(x)-\frac{g_i(x)}{g_{k+i}(x)}g_{(k+i)(k+i+1)}(x))+ \cdots +\linebreak u^{k-1}(g_{i(k-1)}(x)-\frac{g_i(x)}{g_{k+i}(x)}g_{(k+i)(2k-1)}(x))) \in C$. This implies that $g_{ii}(x)-\frac{g_i(x)}{g_{k+i}(x)}g_{(k+i)(k+i)}(x)\in C_{k+i+1}=\langle g_{k+i+1}(x)\rangle$. Hence, $g_{k+i+1}(x)|g_{ii}(x)-\frac{g_i(x)}{g_{k+i}(x)}g_{(k+i)(k+i)}(x)$. That is $g_{k+i+1}(x)|r_{ii}(x)$, where $r_{ii}(x)=g_{ii}(x)-\frac{g_i(x)}{g_{k+i}(x)}g_{(k+i)(k+i)}(x)$.\\
\item From the proof of Condition 7 we have $vA_i-\frac{g_i(x)}{g_{k+i}(x)}A_{k+i}=v(u^i(g_{ii}(x)-\frac{g_i(x)}{g_{k+i}(x)}g_{(k+i)(k+i)}(x))+u^{i+1}(g_{i(i+1)}(x)-\frac{g_i(x)}{g_{k+i}(x)}g_{(k+i)(k+i+1)}(x))+ \cdots +u^{k-1}(g_{i(k-1)}(x)-\frac{g_i(x)}{g_{k+i}(x)}g_{(k+i)(2k-1)}(x)))$. Now $vA_i-\frac{g_i(x)}{g_{k+i}(x)}A_{k+i}-\linebreak \frac{r_{ii}(x)}{g_{k+i+1}(x)}A_{k+i+1}=v(u^{i+1}(g_{i(i+1)}(x)-\frac{g_i(x)}{g_{k+i}(x)}g_{(k+i)(k+i+1)}(x)-\frac{r_{ii}(x)}{g_{k+i+1}(x)}\linebreak g_{(k+i+1)(k+i+1)}(x))+ \cdots +u^{k-1}(g_{i(k-1)}(x)-\frac{g_i(x)}{g_{k+i}(x)}g_{(k+i)(2k-1)}(x)-\frac{r_{ii}(x)}{g_{k+i+1}(x)}\linebreak g_{(k+i+1)(2k-1)}(x)))\in C$. This implies that $g_{i(i+1)}(x)-\frac{g_i(x)}{g_{k+i}(x)}\linebreak g_{(k+i)(k+i+1)}(x)-\frac{r_{ii}(x)}{g_{k+i+1}(x)}g_{(k+i+1)(k+i+1)}(x)\in C_{k+i+2}=\langle g_{k+i+2}(x)\rangle$. Hence, $g_{k+i+2}(x)|g_{i(i+1)}(x)-\frac{g_i(x)}{g_{k+i}(x)}g_{(k+i)(k+i+1)}(x)-\frac{r_{ii}(x)}{g_{k+i+1}(x)}\linebreak g_{(k+i+1)(k+i+1)}(x)$. That is $g_{k+i+2}(x)|r_{i(i+1)}
(x)$, where $r_{i(i+1)}(x)=\linebreak g_{i(i+1)}(x)-\frac{g_i(x)}{
g_{k+i}(x)}g_{(k+i)(k+i+1)}(x)-\frac{r_{ii}(x)}{g_{k+i+1}(x)}g_{(k+i+1)(k+i+1)}(x)$. We have shown for $j=1$. Similarly we can show for other value of $j$.\\
\item For $1\leq i \leq k$, $A_i=u^{i-1}g_i(x)+u^ig_{ii}(x)+u^{i+1}g_{i(i+1)}(x)+ \cdots +u^{k-1}g_{i(k-1)}(x)+v(g_{ik}(x)+ug_{i(k+1)}(x)+u^2g_{i(k+2)}(x)+ \cdots +u^{k-1}g_{i(2k-1)}(x))$ and $A_{k+i}=v(u^{i-1}g_{k+i}(x)+u^ig_{(k+i)(k+i)}(x)+u^{i+1}g_{(k+i)(k+i+1)}(x)+ \cdots +u^{k-1}g_{(k+i)(2k-1)}(x))$. First we prove the condition for $j=1$. For $1\leq i \leq k-1$, We can write for  $\frac{x^n-1}{g_i(x)}A_i-\frac{x^n-1}{g_i(x)}\frac{g_{ii}(x)}{g_{i+1}(x)}A_{i+1}=u^{i+1}(\frac{x^n-1}{g_i(x)}g_{i(i+1)}(x)-\frac{x^n-1}{g_i(x)}\frac{g_{ii}(x)}{g_{i+1}(x)}g_{(i+1)(i+1)}(x))+u^{i+2}(\frac{x^n-1}{g_i(x)}g_{i(i+2)}(x)-\frac{x^n-1}{g_i(x)}\frac{g_{ii}(x)}{g_{i+1}(x)}g_{(i+1)(i+2)}(x))+\cdots+u^{k-1}(\frac{x^n-1}{g_i(x)}g_{i(k-1)}(x)-\frac{x^n-1}{g_i(x)}\frac{g_{ii}(x)}{g_{i+1}(x)}g_{(i+1)(k-1)}(x))+v((\frac{x^n-1}{g_i(x)}g_{ik}(x)-\frac{x^n-1}{g_i(x)}\frac{g_{ii}(x)}{g_{i+1}(x)}g_{(i+1)k}(x))+u(\frac{x^n-1}{g_i(x)}g_{i(k+1)}(x)-\frac{x^n-1}{g_i(x)}\frac{g_{ii}(x)}{g_{
i+1}(x)}g_{(i+1)(k+1)}(x))+\cdots+u^{k-1}(\frac{x^n-1}{g_i(x)}g_{i(2k-1)}(x)-\frac{x^n-1}{g_i(x)}\frac{g_{ii}(x)}{g_{i+1}(x)}g_{(i+1)(2k-1)}(x)))$. Hence, $\frac{x^n-1}{g_i(x)}\linebreak g_{i(i+1)}(x)-\frac{x^n-1}{g_i(x)}\frac{g_{ii}(x)}{g_{i+1}(x)}g_{(i+1)(i+1)}(x) \in C_{i+2}=\langle g_{i+2}(x)\rangle$. Therefore, $g_{i+2}(x)|\frac{x^n-1}{g_i(x)}(g_{i(i+1)}(x)-\frac{g_{ii}(x)}{g_{i+1}(x)}g_{(i+1)(i+1)}(x))$, for $1\leq i \leq k-1$. Similarly by calculating $\frac{x^n-1}{g_k(x)}A_k-\frac{x^n-1}{g_k(x)}\frac{g_{kk}(x)}{g_{k+1}(x)}A_{k+1}$ and $\frac{x^n-1}{g_{k+l}(x)}A_{k+l}-\frac{x^n-1}{g_{k+l}(x)}\linebreak \frac{g_{(k+l)(k+l)}(x)}{g_{k+l+1}(x)}A_{k+l+1}$, for $1\leq l \leq k-1$ we can show $g_{i+2}(x)|\frac{x^n-1}{g_i(x)}(g_{i(i+1)}(x)-\frac{g_{ii}(x)}{g_{i+1}(x)}g_{(i+1)(i+1)}(x))$, for $k\leq i \leq 2k-2$. This implies that $g_{i+2}(x)|\frac{x^n-1}{g_i(x)}\linebreak s_{i(i+1)}(x)$, where $s_{ii}(x)=g_{ii}(x)$ and $s_{i(i+1)}(x)=g_{i(i+1)}(x)-\frac{s_{ii}(x)}{g_{i+1}(x)}\linebreak g_{(i+1)(i+1)}(x)$. Now 
we prove for $j=2$. For $1\leq i \leq k-2$ we can write $\frac{x^n-1}{g_i(x)}A_i-\frac{x^n-1}{g_i(x)}\frac{s_{ii}(x)}{g_{i+1}(x)}A_{i+1}-\frac{x^n-1}{g_i(x)}\frac{s_{i(i+1)}(x)}{g_{i+2}(x)}A_{i+2}=u^{i+2}\frac{x^n-1}{g_i(x)}(g_{i(i+2)}(x)-\frac{s_{ii}(x)}{g_{i+1}(x)}g_{(i+1)(i+2)}(x)- \frac{s_{i(i+1)}}{g_{i+2}(x)}g_{(i+2)(i+2)}(x))$ $+\cdots+$ $u^{k-1}(\frac{x^n-1}{g_i(x)}g_{i(k-1)}(x)-\frac{x^n-1}{g_i(x)}\frac{s_{ii}(x)}{g_{i+1}(x)}g_{(i+1)(k-1)}(x)-\frac{x^n-1}{g_i(x)}\frac{s_{i(i+1)}(x)}{g_{i+2}(x)}g_{(i+2)(k-1)}(x))+v((\frac{x^n-1}{g_i(x)}g_{ik}(x)-\frac{x^n-1}{g_i(x)}\frac{s_{ii}(x)}{g_{i+1}(x)}g_{(i+1)k}(x)- \frac{x^n-1}{g_i(x)}\frac{s_{i(i+1)}(x)}{g_{i+2}(x)}g_{(i+2)k}(x))+u(\frac{x^n-1}{g_i(x)}g_{i(k+1)}(x)-\frac{x^n-1}{g_i(x)}\linebreak \frac{s_{ii}(x)}{g_{i+1}(x)}g_{(i+1)(k+1)}(x)-\frac{x^n-1}{g_i(x)}\frac{s_{i(i+1)}(x)}{g_{i+2}(x)}g_{(i+2)(k+1)}(x))+\cdots+u^{k-1}(\frac{x^n-1}{g_i(x)}\linebreak g_{i(2k-1)}(x)-\frac{x^n-1}{g_i(x)}\frac{s_{ii}(x)}{g_{i+1}(x)}g_{(i+1)(2k-1)}(x)-\frac{x^n-1}{g_i(x)}\
\frac{s_{i(i+1)}(x)}{g_{i+2}(x)}g_{(i+2)(2k-1)}(x)))$.\linebreak Therefore, $\frac{x^n-1}{g_i(x)}(g_{i(i+2)}(x)-\frac{s_{ii}(x)}{g_{i+1}(x)}g_{(i+1)(i+2)}(x)-\frac{s_{i(i+1)}(x)}{g_{i+2}(x)}g_{(i+2)(i+2)}(x))\in C_{i+3}=\langle g_{i+3}(x)\rangle$. Hence, $g_{i+3}(x)|\frac{x^n-1}{g_i(x)}(g_{i(i+2)}(x)-\frac{s_{ii}(x)}{g_{i+1}(x)}g_{(i+1)(i+2)}(x)-\frac{s_{i(i+1)}(x)}{g_{i+2}(x)}g_{(i+2)(i+2)}(x))$ or $g_{i+3}(x)|\frac{x^n-1}{g_i(x)}s_{i(i+2)}(x)$, where $s_{i(i+2)}(x)=\linebreak (g_{i(i+2)}(x)-\frac{s_{ii}(x)}{g_{i+1}(x)}g_{(i+1)(i+2)}(x)-\frac{s_{i(i+1)}(x)}{g_{i+2}(x)}g_{(i+2)(i+2)}(x))$. Similarly by calculating $\frac{x^n-1}{g_{k-1}(x)}A_{k-1}-\frac{x^n-1}{g_{k-1}(x)}\frac{s_{(k-1)(k-1)}(x)}{g_k(x)}A_k-\frac{x^n-1}{g_{k-1}(x)}\frac{s_{(k-1)k}(x)}{g_{k+1}(x)}A_{k+1}$, $\frac{x^n-1}{g_k(x)}A_k-\frac{x^n-1}{g_k(x)}\frac{s_{kk}(x)}{g_{k+1}(x)}A_{k+1}-\frac{x^n-1}{g_k(x)}\frac{s_{k(k+1)}(x)}{g_{k+2}(x)}A_{k+2}$ and $\frac{x^n-1}{g_{k+l}(x)}A_{k+l}-\frac{x^n-1}{g_{k+l}(x)}\linebreak \frac{s_{(k+l)(k+l)}(x)}{g_{
k+l+1}(x)}A_{k+l+1}-\frac{x^n-1}{g_{k+l}(x)}\frac{s_{(k+l)(k+l+1)}(x)}{g_{k+l+2}(x)}
A_{k+l+2}$, for $1\leq l \leq k-2$ we can show $g_{i+3}(x)|\frac{x^n-1}{g_i(x)}s_{i(i+2)}(x)$, for $k-1\leq i \leq 2k-2$. By the same fashion we can prove the Condition 9 for the others value of $j$.
\end{enumerate}
\end{proof}
The following theorem characterizes the free cyclic codes over  $R_{u^k,v^2,p}$.

\begin{theo} \label{freecode}
If $C=\langle A_1, A_2, \cdots, A_{2k}\rangle$ is a cyclic code over the ring $R_{u^k,v^2,p}$, then $C$ is a free cyclic code if and only if $g_1(x)=g_{2k}(x)$. In this case, we have $C=\langle A_1\rangle$ and $A_1|(x^n-1)$ in $R_{u^k,v^2,p}$.
\end{theo}
\begin{proof}
Let $g_1(x)=g_{2k}(x)$. From Condition 2 and condition 3 of Theorem \ref{properties} we have $g_{2k}(x)|g_{2k-1}(x)|$ $\cdots$ $|g_{k+2}(x)|g_{k+1}(x)$, $g_k(x)|g_{k-1}(x)| \cdots |g_2(x)|g_1(x)$ and $g_{k+i}(x)|g_i(x)$, for $1\leq i \leq k$, this gives $g_1(x)=g_2(x)= \cdots =g_{2k}(x)$. Here, we have $\text{Im}\phi=\langle g_1(x)+ug_{11}(x)+\cdots+u^{k-1}g_{1(k-1)}(x), ug_2(x)+u^2g_{22}(x)+\cdots+u^{k-1}g_{2(k-1)}(x), \cdots, u^{k-2}g_{k-1}(x)+u^{k-1}g_{(k-1)(k-1)}(x), u^{k-1}g_k(x)\rangle$ and $\text{ker}\phi=v\langle g_{k+1}(x)+ug_{(k+1)(k+1)}(x)+\cdots+u^{k-1}g_{(k+1)(2k-1)}(x), ug_{k+2}(x)+u^2g_{(k+2)(k+2)}(x)+\cdots+u^{k-1}g_{(k+2)(2k-1)}(x)$, $\cdots, u^{k-2}g_{2k-1}(x)+u^{k-1}g_{(2k-1)(2k-1)}(x)$, $u^{k-1}g_{2k}(x)\rangle$. (See Equation \ref{surj-hom} for the definition of $\phi$). From the Theorem 3.3 of \cite{Aks-Pkk13}, we get that if $g_1(x)=g_k(x)$ then $\text{Im}\phi=\langle g_1(x)+ug_{11}(x)+\cdots+u^{k-1}g_{1(k-1)}(x)\rangle$ and if $g_{k+1}(x)=g_{2k}(x)$ then $\text{ker}\phi=v\langle g_{
k+1}(x)+ug_{(k+1)(k+1)}(x)+\cdots+u^{k-1}g_{(k+1)(2k-1)}(x)\rangle$. Therefore, we can write $C=\langle g_1(x)+ug_{11}(x)+\
\cdots+u^{k-1}g_{1(k-1)}(x)+v(g_{1k}(x)+ug_{1(k+1)}(x)+\cdots+u^{k-1}g_{1(2k-1)}(x)), v(g_{k+1}(x)+ug_{(k+1)(k+1)}(x)+\cdots+u^{k-1}g_{(k+1)(2k-1)}(x))\rangle$. Now we show that $g_{1i}(x)=\linebreak g_{(k+1)(k+i)}(x)$ for $1 \leq i\leq k-1$. We can write $vA_1-\frac{g_1(x)}{g_{k+1}(x)}A_{k+1}=uv(g_{11}(x)-\frac{g_1(x)}{g_{k+1}(x)}g_{(k+1)(k+1)}(x))+u^2v(g_{12}(x)-\frac{g_1(x)}{g_{k+1}(x)}g_{(k+1)(k+2)}(x))+ \cdots +u^{k-1}v(g_{1(k-1)}(x)-\frac{g_1(x)}{g_{k+1}(x)}g_{(k+1)(2k-1)}(x))\in C$. This implies that $g_{11}(x)-\frac{g_1(x)}{g_{k+1}(x)}g_{(k+1)(k+1)}(x)\in C_{k+2}=\langle g_{k+2}(x)\rangle$. This gives $g_{k+2}(x)|(g_{11}(x)-\frac{g_1(x)}{g_{k+1}(x)}g_{(k+1)(k+1)}(x))$. Since $g_1(x)=g_{k+1}(x)=g_{k+2}(x)$, we get $g_1(x)|(g_{11}(x)-g_{(k+1)(k+1)}(x))$. Note that $\text{deg}(g_{11}(x)),\text{deg}(g_{(k+1)(k+1)}(x))<\text{deg}(g_1(x))$, this implies that $g_{11}(x)-\linebreak g_{(k+1)(k+1)}(x)=0$. Therefore, $g_{11}(x)=g_{(k+1)(k+1)}(x)$. Hence, $vA_1-\frac{g_1(x)}{g_{k+1}(x)}A_{k+1}=u^2v(g_{12}(x)-\
\frac{g_1(x)}{g_{k+1}(x)}
g_{(k+1)(
k+2)}(x))+u^3v(g_{13}(x)-\frac{g_1(x)}{g_{k+1}(x)}g_{(k+1)(k+3)}(x))$ $+\cdots+$ $u^{k-1}v(g_{1(k-1)}(x)-\frac{g_1(x)}{g_{k+1}(x)}g_{(k+1)(2k-1)}(x)) \in C$. Again, from the above expression it is easy to see that $g_{12}(x)-\frac{g_1(x)}{g_{k+1}(x)}g_{(k+1)(k+2)}(x)\in C_{k+3}=\langle g_{k+3}(x)\rangle$. In a similar way, as above we can show $g_{12}(x)=g_{(k+1)(k+2)}(x)$. By continuing this way, we can show $g_{1i}(x)=g_{(k+1)(k+i)}(x)$ for $3 \leq i \leq k-1$. This implies that $vA_1=A_{k+1}$. This gives $C=\langle g_1(x)+ug_{11}(x)+\cdots+u^{k-1}g_{1(k-1)}(x)+v(g_{1k}(x)+ug_{1(k+1)}(x)+\cdots+u^{k-1}g_{1(2k-1)}(x))\rangle$ and $C \simeq R_{u^k,v^2,p}^{n-{\rm deg}(g_1(x))}$. Hence, $C$ is a free cyclic code. Conversely, if $C$ is a free cyclic code, we must have $C =\langle g_1(x)+ug_{11}(x)+\cdots+u^{k-1}g_{1(k-1)}(x)+v(g_{1k}(x)+ug_{1(k+1)}(x)+\cdots+u^{k-1}g_{1(2k-1)}(x))\rangle$. Since $u^{k-1}v g_{2k}(x) \in C$, we get $u^{k-1}vg_{2k}(x) = u^{k-1}v \
\alpha g_1(x)$ for some $\
\alpha \in \F_p$. Note that $g_{2k}(x)|g_1(x)$, hence by comparing the coefficients both sides, we get $g_1(x) = g_{2k}(x)$. For the second condition, by division algorithm, we have $x^n-1=A_1q(x)+r(x)$, where $r(x)=0$ or $\text{deg}(r(x))<\text{deg}(g_1(x))$. This implies that $r(x)=(x^n-1)-A_1q(x)\in C$. Note that $A_1$ is the lowest degree polynomial in $C$. So $r(x)=0$. Hence, $A_1|(x^n-1)$ in $R_{u^k,v^2,p}$.
\end{proof}

Note that we get the simpler form for the generators of the cyclic code over $R_{u^k,v^2,p}$, like in the above theorem, if we have $g_1(x) = g_2(x)=\cdots=g_i(x)$, for $2 \leq i \leq 2k-1$ and $g_{2k}(x)=g_{2k-1}(x)=\cdots=g_i(x)$ for $2 \leq i \leq 2k-1$.

\subsection{\rm When $n$ is relatively prime to $p$} $ $

Let $n$ be relatively prime to $p$. If $C=\langle A_1, A_2, \cdots, A_{2k}\rangle$ is a cyclic code of length $n$ over the ring $R_{u^k,v^2,p}$ then we have $\text{Im}\phi=\langle g_1(x)+ug_{11}(x)+\cdots+u^{k-1}g_{1(k-1)}(x), ug_2(x)+u^2g_{22}(x)+\cdots+u^{k-1}g_{2(k-1)}(x), \cdots, u^{k-2}g_{k-1}(x)+u^{k-1}g_{(k-1)(k-1)}(x), u^{k-1}g_k(x)\rangle$ and $\text{ker}\phi=v\langle g_{k+1}(x)+ug_{(k+1)(k+1)}(x)+\cdots+u^{k-1}g_{(k+1)(2k-1)}(x), ug_{k+2}(x)+u^2g_{(k+2)(k+2)}(x)+\cdots+u^{k-1}g_{(k+2)(2k-1)}(x)$, $\cdots,\linebreak u^{k-2}g_{2k-1}(x)+u^{k-1}g_{(2k-1)(2k-1)}(x)$, $u^{k-1}g_{2k}(x)\rangle$. (See Equation \ref{surj-hom} for the definition of $\phi$). Since $n$ is relatively prime to $p$ then from Theorem 3.4 of \cite{Aks-Pkk13}, we have $\text{Im}\phi=\langle g_1(x)+ug_2(x)+\cdots+u^{k-1}g_k(x)\rangle$ and $\text{ker}\phi=v\langle g_{k+1}(x)+ug_{k+2}(x)+\cdots+u^{k-1}g_{2k}(x)\rangle$ with $g_{2k}(x)|g_{2k-1}(x)|$ $\cdots$ $|g_{k+2}(x)|g_{k+1}(x)$ and $g_k(x)|g_{k-1}(x)| \cdots |g_2(x)|g_1(x)$. We 
also have 
the condition $g_{k+i}(x)|g_i(x)$ for $1\leq i \leq k$ from Condition 3 of Theorem \ref{properties}. Therefore, the code $C$ can be written as $C=\langle g_1(x)+ug_2(x)+\cdots+u^{k-1}g_k(x)+v(g_{1k}(x)+ug_{1(k+1)}(x)+\cdots+u^{k-1}g_{1(2k-1)}(x)), v(g_{k+1}(x)+ug_{k+2}(x)+\cdots+u^{k-1}g_{2k}(x))\rangle$ with the same conditions as above on $g_i(x)$'s. From Condition 5 of Theorem \ref{properties}, for $i=1$ and $k\leq j \leq 2k-1$, we get $g_{j+1}(x)|\frac{x^n-1}{g_1(x)}\frac{x^n-1}{g_2(x)} \cdots \frac{x^n-1}{g_j(x)}g_{1j}(x)$. Since $n$ is relatively prime to $p$, $x^n-1$ can be uniquely factored as product of distinct irreducible factors. Therefore, we must have $\text{g.c.d.}\left(g_{j+1}(x), \frac{x^n-1}{g_l(x)}\right) = 1$, for $1\leq l \leq j$. This gives $g_{j+1}(x)|g_{1j}(x)$. From Theorem \ref{unique}, we have $\text{deg}(g_{1j}(x))<\text{deg}(g_{j+1}(x))$, for $k\leq j \leq 2k-1$. This gives $g_{1j}(x)=0$, for $k\leq j \leq 2k-1$. Thus we have proved the following theorem.

\begin{theo}\label{relatively-prime}
Let $C=\langle A_1, A_2, \cdots, A_{2k}\rangle$ be a cyclic code over the ring $R_{u^k,v^2,p}$ of length $n$. If $n$ is relatively prime to $p$, then we have $C=\langle g_1(x)+ug_2(x)+\cdots+u^{k-1}g_k(x), v(g_{k+1}(x)+ug_{k+2}(x)+\cdots+u^{k-1}g_{2k}(x))\rangle$ with the condition $g_{2k}(x)|g_{2k-1}(x)|$ $\cdots$ $|g_{k+2}(x)|g_{k+1}(x)$, $g_k(x)|g_{k-1}(x)| \cdots |g_2(x)|g_1(x)$ and $g_{k+i}(x)|g_i(x)$ for $1\leq i \leq k$.
\end{theo}

\section{Ranks and minimal spanning sets}\label{rank}
In this section, we find the rank and minimal spanning set of a cyclic code $C$ over the ring $R_{u^k,v^2,p}$. We follow Dougherty and Shiromoto \cite[page 401]{Dou-Shiro01} for the definition of the rank of a code $C$. We first prove the number of lemmas that we use to find the rank and minimal spanning set of these cyclic codes.\\\\

Let $C=\langle A_1, A_2, \cdots, A_{2k}\rangle$ be a cyclic code over the ring $R_{u^k,v^2,p}$ (see page \pageref{A_i's} for $A_i$'s). We know that $C_i=\langle g_i(x)\rangle$, for $1 \leq i \leq 2k$ (see page \pageref{C_i's}). Let  $t_i={\rm deg}(g_i(x))$ for $1 \leq i \leq 2k$. From Conditions 2 and Condition 3 of Theorem \ref{properties}, we have $g_{2k}(x)|g_{2k-1}(x)| \cdots |g_{k+2}(x)|g_{k+1}(x)$, $g_k(x)|g_{k-1}(x)| \cdots |g_2(x)|g_1(x)$ and $g_{k+i}(x)|g_i(x)$ for $1\leq i \leq k$. Therefore, we get $t_1 \geq t_2 \geq \cdots \geq t_k$, $t_{k+1} \geq t_{k+2} \geq \cdots \geq t_{2k}$ and $t_i \geq t_{k+i}$ for $1\leq i \leq k$.

\begin{lemma} \label{lm-g-p}
Let $C=\langle A_1, A_2, \cdots, A_{2k}\rangle$ be a cyclic code over the ring $R_{u^k,v^2,p}$.  For $1 \leq i \leq k$, we get the following:
\begin{enumerate}
\item Any polynomial in $C$ of the form $v(u^{i-1}p_0(x)+u^ip_1(x)+u^{i+1}p_2(x)+ \cdots +u^{k-1}p_{k-i}(x))$ can be written as $q_0(x)A_{k+i}+q_1(x)A_{k+i+1}+ \cdots +q_{k-i-1}(x)A_{2k-1}+q_{k-i}(x)A_{2k}$ and 
\item any polynomial in $C$ of the form $u^{i-1}p_0(x)+u^ip_1(x)+u^{i+1}p_2(x)+ \cdots +u^{k-1}p_{k-i}(x)+v(p_{k+1}(x)+up_{k+2}(x)+u^2p_{k+3}(x)+ \cdots +u^{k-1}p_{2k}(x))$ can be written as $q_0(x)A_i+q_1(x)A_{i+1}+ \cdots +q_{k-i-1}(x)A_{k-1}+q_{k-i}(x)A_k+q_{k+1}(x)A_{k+1}+q_{k+2}(x)A_{k+2}+q_{k+3}(x)A_{k+3}+ \cdots +q_{2k-1}(x)A_{2k-1}+q_{2k}(x)A_{2k}$
\end{enumerate}
for some $q_0(x),q_1(x), \cdots, q_{k-i-1}(x), q_{k-i}(x), q_{k+1}(x), \cdots, q_{2k}(x)\in \F_p[x]$.
\end{lemma}
\begin{proof}
$(1)$ Let $A'=v(u^{i-1}p_0(x)+u^ip_1(x)+u^{i+1}p_2(x)+ \cdots +u^{k-1}p_{k-i}(x))\in C$. This implies that $p_0(x)\in C_{k+i}= \langle g_{k+i}(x) \rangle$. That is $g_{k+i}(x)|p_0(x)$, thus, $p_0(x)=q_0(x)g_{k+i}(x)$ for some $q_0(x) \in \F_p[x]$. We have 
\begin{multline}
A_{k+i}=v(u^{i-1}g_{k+i}(x)+u^ig_{(k+i)(k+i)}(x)\\+u^{i+1}g_{(k+i)(k+i+1)}(x)+ \cdots +u^{k-1}g_{(k+i)(2k-1)}(x)).\notag
\end{multline}
Therefore,
\begin{multline}
A'-q_0(x)A_{k+i}=v(u^i(p_1(x)-q_0(x)g_{(k+i)(k+i)}(x))\\+u^{i+1}(p_2(x)-q_0(x)g_{(k+i)(k+i+1)}(x))+u^{i+2}(p_3(x)-q_0(x)g_{(k+i)(k+i+2)}(x))\\+ \cdots +u^{k-1}(p_{k-i}(x)-q_0(x)g_{(k+i)(2k-1)}(x)).\notag
\end{multline}
Since $A', A_{k+i} \in C$ and $C$ is an ideal, we get $A'-q_0(x)A_{k+i} \in C$.
This implies that
\begin{equation}
p_1(x)-q_0(x)g_{(k+i)(k+i)}(x) \in C_{k+i+1}= \langle g_{k+i+1}(x) \rangle.\notag
\end{equation}
Thus,
\begin{equation}
g_{k+i+1}(x)|(p_1(x)-q_0(x)g_{(k+i)(k+i)}(x)).\notag
\end{equation}
Therefore,
\begin{equation}
p_1(x)-q_0(x)g_{(k+i)(k+i)}(x)=q_1(x)g_{k+i+1}(x)\notag
\end{equation}
for some $q_1(x) \in \F_p[x]$. Again,
\begin{multline}
A'-q_0(x)A_{k+i}-q_1(x)A_{k+i+1}\\=v(u^{i+1}(p_2(x)-q_0(x)g_{(k+i)(k+i+1)}(x)-q_1(x)g_{(k+i+1)(k+i+1)}(x))\\+u^{i+2}(p_3(x)-q_0(x)g_{(k+i)(k+i+2)}(x)-q_1(x)g_{(k+i+1)(k+i+2)}(x))\\+ \cdots +u^{k-1}(p_{k-i}(x)-q_0(x)g_{(k+i)(2k-1)}(x)-q_1(x)g_{(k+i+1)(2k-1)}(x)).\notag
\end{multline}
Proceeding in this way, after $k-i$ times we get
\begin{multline}
A'-q_0(x)A_{k+i}-q_1(x)A_{k+i+1}- \cdots -q_{k-i-1}(x)A_{2k-1}\\=vu^{k-1}(p_{k-i}(x)-q_0(x)g_{(k+i)(2k-1)}(x)-q_1(x)g_{(k+i+1)(2k-1)}(x)\\- \cdots -q_{k-i-1}(x)g_{(2k-1)(2k-1)}(x))\in C,\notag
\end{multline}
for some $q_2(x), q_3(x), \cdots, q_{k-i-1}(x) \in \F_p[x]$.
This implies that
\begin{multline}
p_{k-i}(x)-q_0(x)g_{(k+i)(2k-1)}(x)-q_1(x)g_{(k+i+1)(2k-1)}(x)\\- \cdots -q_{k-i-1}(x)g_{(2k-1)(2k-1)}(x)\in C_{2k}= \langle g_{2k}(x) \rangle.\notag
\end{multline}
That is
\begin{multline}
p_{k-i}(x)-q_0(x)g_{(k+i)(2k-1)}(x)-q_1(x)g_{(k+i+1)(2k-1)}(x)\\- \cdots -q_{k-i-1}(x)g_{(2k-1)(2k-1)}(x)=q_{k-i}(x)g_{2k}(x)\notag
\end{multline}
for some $q_{k-i}(x) \in \F_p[x]$. Therefore,
\begin{equation}
A'-q_0(x)A_{k+i}-q_1(x)A_{k+i+1}- \cdots -q_{k-i-1}(x)A_{2k-1}-q_{k-i}(x)A_{2k}=0.\notag
\end{equation}
This gives,
\begin{multline}
A'=q_0(x)A_{k+i}+q_1(x)A_{k+i+1}+ \cdots +q_{k-i-1}(x)A_{2k-1}+q_{k-i}(x)A_{2k}.\notag
\end{multline}
This proves Statement $1$.\\ 
$(2)$ The proof is similar to $1$.
\end{proof}
This lemma is referred in the proof of Lemma \ref{lm-t_k+i-1-t_k+i} and Lemma \ref{lm-t_i-t_k+i}.

\begin{remark} \label{rank-remark}
Note that in the proof of Statement 1 of Lemma \ref{lm-g-p}, we have the following relation between $p_j(x)$ and $q_j(x)$'s, for $0 \leq j \leq k-i$ and $1 \leq i \leq k$ : $p_0(x)=q_0(x)g_{k+i}(x)$, $p_1(x)-q_0(x)g_{(k+i)(k+i)}(x)=q_1(x)g_{k+i+1}(x)$, $\cdots$ and $p_{k-i}(x)-q_0(x)g_{(k+i)(2k-1)}(x)-q_1(x)g_{(k+i+1)(2k-1)}(x)- \cdots -q_{k-i-1}(x)\linebreak g_{(2k-1)(2k-1)}(x)=q_{k-i}(x)g_{2k}(x)$.
\end{remark}
This remark is referred in the proof of Lemma \ref{lm-t_k+i-1-t_k+i}.

\begin{lemma}\label{lm-t_k+i-1-t_k+i}
Let $C=\langle A_1, A_2, \cdots, A_{2k}\rangle$ be a cyclic code over the ring $R_{u^k,v^2,p}$.  $x^{t_{k+i-1}-t_{k+i}}A_{k+i}$, for $2 \leq i \leq k$, can be written as $x^{t_{k+i-1}t_{k+i}}A_{k+i}$ $=$ $c_{k+i-1}uA_{k+i-1}$ $+$ $q_0(x)A_{k+i}$ $+$ $q_1(x)A_{k+i+1}+q_2(x)A_{k+i+2}+ \cdots +q_{k-i}(x)A_{2k}$, for some $c_{k+i-1}\in \F_p$ and $q_0(x),q_1(x), \cdots, q_{k-i}(x)\in \F_p[x]$ with $\text{deg}(q_0(x))<t_{k+i-1}-t_{k+i}$,\linebreak $\text{deg}(q_1(x))<(t_{k+i-1}-t_{k+i}) ~ or ~ (t_{k+i}-t_{k+i+1})$, $\text{deg}(q_2(x))<(t_{k+i-1}-t_{k+i}) ~ or\linebreak (t_{k+i}-t_{k+i+1}) ~ or ~ (t_{k+i+1}-t_{k+i+2})$, $\cdots$ and $\text{deg}(q_{k-i}(x))<(t_{k+i-1}-t_{k+i}) ~ or\linebreak (t_{k+i}-t_{k+i+1}) ~or ~ (t_{k+i+1}-t_{k+i+2}) ~ or ~ \cdots ~or ~ (t_{2k-1}-t_{2k})$.
\end{lemma}
\begin{proof}
We have
\begin{multline}
A_{k+i}=v(u^{i-1}g_{k+i}(x)+u^ig_{(k+i)(k+i)}(x)\\+u^{i+1}g_{(k+i)(k+i+1)}(x)+ \cdots +u^{k-1}g_{(k+i)(2k-1)}(x))\\=u^{i-1}v(g_{k+i}(x)+ug_{(k+i)(k+i)}(x)+u^2g_{(k+i)(k+i+1)}(x)\\+ \cdots +u^{k-i}g_{(k+i)(2k-1)}(x))\notag
\end{multline}
and
\begin{multline}
A_{k+i-1}=u^{i-2}v(g_{k+i-1}(x)+ug_{(k+i-1)(k+i-1)}(x)\\+u^2g_{(k+i-1)(k+i)}(x)+ \cdots +u^{k-i+1}g_{(k+i-1)(2k-1)}(x)).\notag
\end{multline}
Since the above two polynomials $g_{k+i}(x)+ug_{(k+i)(k+i)}(x)+u^2g_{(k+i)(k+i+1)}(x)+ \cdots +u^{k-i}g_{(k+i)(2k-1)}(x)$ and $g_{k+i-1}(x)+ug_{(k+i-1)(k+i-1)}(x)+u^2g_{(k+i-1)(k+i)}(x)+ \cdots +u^{k-i+1}g_{(k+i-1)(2k-1)}(x)$ are regular, from Proposition \ref{division-alg} we can apply the division algorithm for these two polynomial. Let the leading coefficient of $x^{t_{k+i-1}-t_{k+i}}(g_{k+i}(x)+ug_{(k+i)(k+i)}(x)+ \cdots + u^{k-i}g_{(k+i)(2k-1)}(x))$ be $\alpha_i$ and of $g_{k+i-1}(x)+ug_{(k+i-1)(k+i-1)}(x)+ \cdots + u^{k-i+1}g_{(k+i-1)(2k-1)}(x)$ be $\beta_{i-1}$. There exists a constant $c_{k+i-1}\in \F_p$ such that $\alpha_i=c_{k+i-1}\beta_{i-1}$. By the division algorithm, we have
\begin{multline} \label{eq-t_k+i-1-t_k+i-div-alg}
x^{t_{k+i-1}-t_{k+i}}(g_{k+i}(x)+ug_{(k+i)(k+i)}(x)+ \cdots +u^{k-i}g_{(k+i)(2k-1)}(x))\\=c_{k+i-1}(g_{k+i-1}(x)
+ug_{(k+i-1)(k+i-1)}(x)+ \cdots +u^{k-i+1}g_{(k+i-1)(2k-1)}(x))\\+(p_0(x)+up_1(x)+ \cdots +u^{k-i+1}(x)p_{k-i+1}(x)),
\end{multline}
where $p_0(x)+up_1(x)+ \cdots +u^{k-i+1}(x)p_{k-i+1}(x)$ is the remainder term and $\text{deg}(p_0(x))<\text{deg}(g_{k+i-1}(x))=t_{k+i-1}$. This gives by comparing coefficient $p_0(x)=x^{t_{k+i-1}-t_{k+i}}g_{k+i}(x)-c_{k+i-1}g_{k+i-1}(x)$, $p_1(x)=x^{t_{k+i-1}-t_{k+i}}g_{(k+i)(k+i)}(x)-c_{k+i-1}g_{(k+i-1)(k+i-1)}(x)$, $\cdots$, $p_{k-i}(x)=x^{t_{k+i-1}-t_{k+i}}g_{(k+i)(2k-1)}(x)-c_{k+i-1}\linebreak g_{(k+i-1)(2k-2)}(x)$ and $p_{k-i+1}(x)=-c_{k+i-1}g_{(k+i-1)(2k-1)}(x)$. Multiplying both side of Equation \ref{eq-t_k+i-1-t_k+i-div-alg} by $u^{i-1}v$ gives,
\begin{multline}
x^{t_{k+i-1}-t_{k+i}}v(u^{i-1}g_{k+i}(x)+u^ig_{(k+i)(k+i)}(x)+ \cdots +u^{k-1}g_{(k+i)(2k-1)}(x))\\=c_{k+i-1}vu(u^{i-2}g_{k+i-1}(x)+u^{i-1}g_{(k+i-1)(k+i-1)}(x)+ \cdots +u^{k-1}g_{(k+i-1)(2k-1)}(x))\\+v(u^{i-1}p_0(x)\linebreak +u^ip_1(x)+ \cdots +u^{k-1}(x)p_{k-i}(x)).\notag
\end{multline}
This can be written as
\begin{multline}
x^{t_{k+i-1}-t_{k+i}}A_{k+i}=c_{k+i-1}uA_{k+i-1}\\+v(u^{i-1}p_0(x)+u^ip_1(x)+ \cdots +u^{k-1}(x)p_{k-i}(x)).\notag
\end{multline}
That is
\begin{multline}
x^{t_{k+i-1}-t_{k+i}}A_{k+i}-c_{k+i-1}uA_{k+i-1}\\=v(u^{i-1}p_0(x)+u^ip_1(x)+\cdots +u^{k-1}(x)p_{k-i}(x)).\notag
\end{multline}
Now $x^{t_{k+i-1}-t_{k+i}}A_{k+i}$ and $c_{k+i-1}uA_{k+i-1}\in C$, this implies that
\begin{multline}
x^{t_{k+i-1}-t_{k+i}}A_{k+i}-c_{k+i-1}uA_{k+i-1}\\=v(u^{i-1}p_0(x)+u^ip_1(x)+ \cdots +u^{k-1}(x)p_{k-i}(x))\in C.\notag
\end{multline}
From Statement 1 of Lemma \ref{lm-g-p}, we can write
\begin{multline}
v(u^{i-1}p_0(x)+u^ip_1(x)+ \cdots +u^{k-1}(x)p_{k-i}(x))\\=q_0(x)A_{k+i}+q_1(x)A_{k+i+1}+ \cdots +q_{k-i}(x)A_{2k}.\notag
\end{multline}
for some $q_0(x)$, $q_1(x)$, $\cdots$, $q_{k-i}(x)\in \F_p[x]$. Hence,
\begin{multline}
x^{t_{k+i-1}-t_{k+i}}A_{k+i}-c_{k+i-1}uA_{k+i-1}=q_0(x)A_{k+i}+q_1(x)A_{k+i+1}+ \cdots +q_{k-i}(x)A_{2k},\notag
\end{multline}
that is,
\begin{multline}
x^{t_{k+i-1}-t_{k+i}}A_{k+i}=c_{k+i-1}uA_{k+i-1}+q_0(x)A_{k+i}+q_1(x)A_{k+i+1}+ \cdots +q_{k-i}(x)A_{2k}.\notag
\end{multline}
This prove the first part of the lemma. Now we prove the degree result. To show the degree results, from Remark \ref{rank-remark}, we have
\begin{equation}
p_0(x)=q_0(x)g_{k+i}(x),\notag
\end{equation}
this implies that
\begin{equation}
\text{deg}(p_0(x))=\text{deg}(q_0(x))+\text{deg}(g_{k+i}(x)).\notag
\end{equation}
Therefore,
\begin{equation}
\text{deg}(p_0(x))=\text{deg}(q_0(x))+t_{k+i}.\notag
\end{equation}
From Equation \ref{eq-t_k+i-1-t_k+i-div-alg}, $\text{deg}(p_0(x))<t_{k+i-1}$, thus
\begin{equation}
\text{deg}(q_0(x))<t_{k+i-1}-t_{k+i}.\notag
\end{equation}
Again, from Equation \ref{eq-t_k+i-1-t_k+i-div-alg}, we have 
\begin{equation} \label{eq-p_1-a}
p_1(x)=x^{t_{k+i-1}-t_{k+i}}g_{(k+i)(k+i)}(x)-c_{k+i-1}g_{(k+i-1)(k+i-1)}(x).
\end{equation}
Also, from Remark \ref{rank-remark}, we have
\begin{equation} \label{eq-p_1-b}
p_1(x)-q_0(x)g_{(k+i)(k+i)}(x)=q_1(x)g_{k+i+1}(x).
\end{equation}
Therefore, Equations \ref{eq-p_1-a} and Equation \ref{eq-p_1-b} gives
\begin{multline} \label{eq-deg-p_1}
x^{t_{k+i-1}-t_{k+i}}g_{(k+i)(k+i)}(x)-c_{k+i-1}g_{(k+i-1)(k+i-1)}(x)-q_0(x)g_{(k+i)(k+i)}(x)\\=q_1(x)g_{k+i+1}(x).
\end{multline}
The degree of the polynomial on right hand side of Equation \ref{eq-deg-p_1} is less than or equal to the degree of highest degree polynomial on the left hand side. We have
\begin{equation}
\text{deg}(q_0(x)g_{(k+i)(k+i)}(x))<\text{deg}(x^{t_{k+i-1}-t_{k+i}}g_{(k+i)(k+i)}(x)),\notag
\end{equation}
because $\text{deg}(q_0(x))<t_{k+i-1}-t_{k+i}$. Thus, $q_0(x)g_{(k+i)(k+i)}(x)$ is not the highest degree polynomial in the left hand side of Equation \ref{eq-deg-p_1}. Therefore, either polynomial $x^{t_{k+i-1}-t_{k+i}}g_{(k+i)(k+i)}(x)$ or polynomial $c_{k+i-1}g_{(k+i-1)(k+i-1)}(x)$ is highest degree polynomial or both has equal degree. If $x^{t_{k+i-1}-t_{k+i}}g_{(k+i)(k+i)}(x)$ is the highest degree polynomial. From Equation \ref{eq-deg-p_1} we have
\begin{multline}
\text{deg}(q_1(x)g_{k+i+1}(x)) \leq \text{deg}(x^{t_{k+i-1}-t_{k+i}}g_{(k+i)(k+i)}(x))\\=t_{k+i-1}-t_{k+i}+\text{deg}(g_{(k+i)(k+i)}(x))\notag
\end{multline}
From Theorem \ref{unique}, we have $\text{deg}(g_{(k+i)(k+i)}(x))<t_{k+i+1}$. Therefore,
\begin{equation}
\text{deg}(q_1(x)g_{k+i+1}(x))<t_{k+i-1}-t_{k+i}+t_{k+i+1}\notag
\end{equation}
This gives
\begin{equation}
\text{deg}(q_1(x))+t_{k+i+1}<t_{k+i-1}-t_{k+i}+t_{k+i+1}.\notag
\end{equation}
This implies that
\begin{equation}
\text{deg}(q_1(x))<t_{k+i-1}-t_{k+i}.\notag
\end{equation}
Again, If $c_{k+i-1}g_{(k+i-1)(k+i-1)}(x)$ is the highest degree polynomial. From Equation \ref{eq-deg-p_1}, we have 
\begin{equation}
\text{deg}(q_1(x)g_{k+i+1}(x))\leq \text{deg}(c_{k+i-1}g_{(k+i-1)(k+i-1)}(x))\notag
\end{equation}
From Theorem \ref{unique}, we have $\text{deg}(g_{(k+i-1)(k+i-1)}(x))<t_{k+i}$. This gives,
\begin{equation}
\text{deg}(q_1(x))+t_{k+i+1}<t_{k+i}.\notag
\end{equation}
This implies that
\begin{equation}
\text{deg}(q_1(x))<t_{k+i}-t_{k+i+1}.\notag
\end{equation}
That is finally we get
\begin{equation}
\text{deg}(q_1(x))<(t_{k+i-1}-t_{k+i}) ~ or ~ (t_{k+i}-t_{k+i+1}).\notag
\end{equation}
Proceeding  in a similar way we can prove that $\text{deg}(q_2(x))<(t_{k+i-1}-t_{k+i})$ or $(t_{k+i}-t_{k+i+1})$ or $(t_{k+i+1}-t_{k+i+2}), \cdots$ and $\text{deg}(q_{k-i}(x))<(t_{k+i-1}-t_{k+i})$ or $(t_{k+i}-t_{k+i+1})$ or $(t_{k+i+1}-t_{k+i+2})$ or $\cdots$ or $(t_{2k-1}-t_{2k})$.
\end{proof}
This lemma is referred in the proof of Case \ref{case-1} and Case \ref{case-3} in Theorem \ref{rank-main}.

\begin{lemma}\label{lm-t_i-t_k+i}
Let $C=\langle A_1, A_2, \cdots, A_{2k}\rangle$ be a cyclic code over the ring $R_{u^k,v^2,p}$. $x^{t_i-t_{k+i}}A_{k+i}$, for $1 \leq i \leq k$, can be written as $x^{t_i-t_{k+i}}A_{k+i}=c_ivA_i+q_0(x)A_{k+i}+q_1(x)A_{k+i+1}+q_2(x)A_{k+i+2}+ \cdots +q_{k-i}(x)A_{2k}$ for some $c_i\in \F_p$ and $q_0(x),q_1(x)$, $\cdots, q_{k-i}(x)\in \F_p[x]$ with $\text{deg}(q_0(x))<t_i-t_{k+i}$, $\text{deg}(q_1(x))<(t_i-t_{k+i}) ~ or ~ (t_{k+i}-t_{k+i+1})$, $\text{deg}(q_2(x))<(t_{i}-t_{k+i}) ~ or ~ (t_{k+i}-t_{k+i+1}) ~ or ~ (t_{k+i+1}-t_{k+i+2}), \cdots$ and $\text{deg}(q_{k-i}(x))<(t_{i}-t_{k+i}) ~ or ~ (t_{k+i}-t_{k+i+1}) ~ or ~ (t_{k+i+1}-t_{k+i+2}) ~ or ~ \cdots ~or ~ (t_{2k-1}-t_{2k})$.
\end{lemma}
\begin{proof}
From Condition 3 of Theorem \ref{properties}, we have $g_{k+i}(x)|g_i(x)$ for $1 \leq i \leq k$. This implies that $g_i(x)=s_i(x)g_{k+i}(x)$ for some $s_i(x) \in \F_p[x]$. This can be written as $g_i(x)=g_{k+i}(x)(s_{i0}+s_{i1}x +\cdots +s_{i(t_i-t_{k+i})}x^{t_i-t_{k+i}})$, where $s_{i0}, s_{i1}, \cdots,\linebreak s_{i(t_i-t_{k+i})} \in \F_p$. Clearly, $s_{i(t_i-t_{k+i})} \neq 0$ (by degree comparison). For $1 \leq i \leq k$, we have
\begin{multline}
A_i=u^{i-1}g_i(x)+u^ig_{ii}(x)+u^{i+1}g_{i(i+1)}(x)+ \cdots +u^{k-1}g_{i(k-1)}(x)\\+v(g_{ik}(x)+ug_{i(k+1)}(x)+u^2g_{i(k+2)}(x)+ \cdots +u^{k-1}g_{i(2k-1)}(x))\notag
\end{multline}
and
\begin{multline}
A_{k+i}=v(u^{i-1}g_{k+i}(x)+u^ig_{(k+i)(k+i)}(x)\\+u^{i+1}g_{(k+i)(k+i+1)}(x)+ \cdots +u^{k-1}g_{(k+i)(2k-1)}(x)).\notag
\end{multline}
Therefore,
\begin{multline}
vA_i-s_i(x)A_{k+i}=v(u^i(g_{ii}(x)-s_i(x)g_{(k+i)(k+i)}(x))+u^{i+1}(g_{i(i+1)}(x)\\-s_i(x)g_{(k+i)(k+i+1)}(x))+ \cdots +u^{k-1}(g_{i(k-1)}(x)-s_i(x)g_{(k+i)(2k-1)}(x))).\notag
\end{multline}
Again, $vA_i-s_i(x)A_{k+i}\in C$. From Statement 1 of Lemma \ref{lm-g-p}, we have
\begin{multline}
v(u^i(g_{ii}(x)-s_i(x)g_{(k+i)(k+i)}(x))+u^{i+1}(g_{i(i+1)}(x)\\-s_i(x)g_{(k+i)(k+i+1)}(x))+ \cdots +u^{k-1}(g_{i(k-1)}(x)-s_i(x)g_{(k+i)(2k-1)}(x)))\\=q_1'(x)A_{k+i+1}+q_2'(x)A_{k+i+2}+ \cdots +q_{k-i}'(x)A_{2k}.\notag
\end{multline}
for some $q_1'(x)$, $q_2'(x)$, $\cdots$, $q_{k-i}'(x)\in \F_p[x]$. That is
\begin{multline}
vA_i-s_i(x)A_{k+i}=q_1'(x)A_{k+i+1}+q_2'(x)A_{k+i+2}+ \cdots +q_{k-i}'(x)A_{2k},\notag
\end{multline}
thus
\begin{multline}
s_i(x)A_{k+i}=vA_i-q_1'(x)A_{k+i+1}-q_2'(x)A_{k+i+2}- \cdots -q_{k-i}'(x)A_{2k},\notag
\end{multline}
this can be written as
\begin{multline}
(s_{i0}+s_{i1}x +\cdots +s_{i(t_i-t_{k+i})}x^{t_i-t_{k+i}})A_{k+i}\\=vA_i-q_1'(x)A_{k+i+1}-q_2'(x)A_{k+i+2}- \cdots -q_{k-i}'(x)A_{2k}.\notag
\end{multline}
This gives,
\begin{multline}
s_{i(t_i-t_{k+i})}x^{t_i-t_{k+i}}A_{k+i}=vA_i-(s_{i0}+s_{i1}x +\cdots +s_{i(t_i-t_{k+i}-1)}x^{t_i-t_{k+i}-1})A_{k+i}\\-q_1'(x)A_{k+i+1}-q_2'(x)A_{k+i+2}- \cdots -q_{k-i}'(x)A_{2k}.\notag
\end{multline}
Therefore,
\begin{multline}
s_{i(t_i-t_{k+i})}x^{t_i-t_{k+i}}A_{k+i}=vA_i-q_0'(x)A_{k+i}\\-q_1'(x)A_{k+i+1}-q_2'(x)A_{k+i+2}- \cdots -q_{k-i}'(x)A_{2k},\notag
\end{multline}
where $q_0'(x)=s_{i0}+s_{i1}x +\cdots +s_{i(t_i-t_{k+i}-1)}x^{t_i-t_{k+i}-1}$. This implies that
\begin{multline}
x^{t_i-t_{k+i}}A_{k+i}=s_{i(t_i-t_{k+i})}^{-1}vA_i-s_{i(t_i-t_{k+i})}^{-1}q_0'(x)A_{k+i}-s_{i(t_i-t_{k+i})}^{-1}q_1'(x)A_{k+i+1}\\-s_{i(t_i-t_{k+i})}^{-1}q_2'(x)A_{k+i+2}- \cdots -s_{i(t_i-t_{k+i})}^{-1}q_{k-i}'(x)A_{2k}.\notag
\end{multline}
This equation can be written as
\begin{multline}
x^{t_i-t_{k+i}}A_{k+i}=c_ivA_i+q_0(x)A_{k+i}\\+q_1(x)A_{k+i+1}+q_2(x)A_{k+i+2}+ \cdots +q_{k-i}(x)A_{2k},\notag
\end{multline}
where $c_i=s_{i(t_i-t_{k+i})}^{-1}$, $q_0(x)=-s_{i(t_i-t_{k+i})}^{-1}q_0'(x)$, $\cdots$, $q_{k-i}(x)=-s_{i(t_i-
t_{k+i})}^{-1}q_{k-i}'(x)$. The proof of the degree results are same as shown in Lemma \ref{lm-t_k+i-1-t_k+i}.
\end{proof}
This lemma is referred in the proof of Case \ref{case-1}, Case \ref{case-2} and Case \ref{case-3} in Theorem \ref{rank-main}.

\begin{lemma}\label{lm-vA_i}
Let $C$ be a cyclic code over the ring $R_{u^k,v^2,p}$. If $C=\langle vA_1, vA_2,$ $\cdots,$ $vA_{k}\rangle$, then the spanning set $S$ of the code  $C$ is $\{vA_1, xvA_1$, $\cdots$, $x^{n-t_1-1}vA_1$, $A_2, xvA_2$, $\cdots$, $x^{t_1-t_2-1}vA_2,$ $\cdots$, $vA_k, xvA_k$,$\cdots$, $x^{t_{k-1}-t_k-1}vA_k$\}.
\end{lemma}
\begin{proof}
It is suffices to show that $S$ spans the set $S'=\{vA_1, xvA_1$, $\cdots$, $x^{n-t_1-1}vA_1$, $A_2, xvA_2$, $\cdots$, $x^{n-t_2-1}vA_2$, $\cdots$, $vA_{k}, xvA_{k}$, $\cdots$, $x^{n-t_{k}-1}vA_{k}\}$. We can also visualize $C$ as a $R_{u^k,p}$-module, hence as a cyclic code over the ring $R_{u^k,p}$. Therefore, from Theorem 4.2 of \cite{Aks-Pkk13}, it is easy to see that, the elements of the set $\{x^{t_{k-1}-t_{k}}vA_{k}$, $x^{t_{k-1}-t_{k}+1}vA_{k}$, $\cdots$, $x^{n-t_{k}-1}vA_{k}$, $x^{t_{k-2}-t_{k-1}}vA_{k-1}$, $x^{t_{k-2}-t_{k-1}+1}vA_{k-1}$, $\cdots$, $x^{n-t_{k-1}-1}vA_{k-1}$, $\cdots$,
$x^{t_{1}-t_{2}}vA_{2}$, $x^{t_{1}-t_{2}+1}A_{2}$, $\cdots$,\linebreak $x^{n-t_{2}-1}A_{2}$\}
can be written as $R_{u^k,p}$-linear combination of the elements of the set \{$vA_{1}$, $xvA_{1}$, $\cdots$, $x^{n-t_{1}-1}vA_{1}$, $vA_{2}$, $xvA_{2}$, $\cdots$, $x^{t_{1}-t_{2}-1}vA_{2}$, $\cdots$, $vA_{k}$, $xvA_{k}$, $\cdots$, $x^{t_{k-1}-t_{k}-1}vA_{k}$\}. This proves the lemma.
\end{proof}
Note that $\langle vA_1, vA_2, \cdots, vA_{k}\rangle$ is a subcode of the code $\langle A_1, A_2, \cdots, A_{2k}\rangle$.
This lemma is referred in the proof of Case \ref{case-2} in Theorem \ref{rank-main}.

\begin{theo} \label{rank-main}
Let $n$ be a positive integer not relatively prime to $p$ and $C$ be a cyclic code of length $n$ over the ring $R_{u^k,v^2,p}$. If $C=\langle A_1,A_2, \cdots ,A_{2k}\rangle$, ${\rm deg}(g_i(x))=t_i$, $1 \leq i \leq 2k$ and $t_i' = {\rm min}\{{\rm deg}(g_{i+1}(x)),{\rm deg}(g_{k+i}(x))\}$, $1 \leq i \leq k-1$, then $C$ has free rank $n-t_1$ and rank $n+t_1+t_1^{'}+t_2^{'}+\cdots+t_{k-1}^{'}-(t_k+t_{k+1}+\cdots+t_{2k})$. The minimal spanning set $B$ of the code  $C$ is $\{A_1, xA_1$, $\cdots$, $x^{n-t_1-1}A_1$, $A_2, xA_2$, $\cdots$, $x^{t_1-t_2-1}A_2, A_3, xA_3$, $\cdots$, $x^{t_2-t_3-1}A_3$, $\cdots$, $A_k, xA_k$, $\cdots$, $x^{t_{k-1}-t_k-1}A_k, A_{k+1}, xA_{k+1}$, $\cdots$, $x^{t_1-t_{k+1}-1}A_{k+1}$, $A_{k+2}$, $xA_{k+2}$, $\cdots$,\linebreak $x^{t_1'-t_{k+2}-1}A_{k+2}$, $A_{k+3}$, $xA_{k+3}$, $\cdots$, $x^{t_2'-t_{k+3}-1}A_{k+3}$, $\cdots$, $A_{2k-1}$, $xA_{2k-1}$, $\cdots$, $x^{t_{k-2}'-t_{2k-1}-1}A_{2k-1}$, $A_{2k}, xA_{2k}$, $\cdots$, $x^{t_{k-1}'-t_{2k}-1}A_{2k}\}$.
\end{theo}
\begin{proof}
It is suffices to show that $B$ spans the set $B'=\{A_1$, $xA_1$, $\cdots$, $x^{n-t_1-1}A_1$, $A_2$, $xA_2$, $\cdots$, $x^{n-t_2-1}A_2$, $\cdots$, $A_k$, $xA_k$, $\cdots$, $x^{n-t_k-1}A_k$, $A_{k+1}$, $xA_{k+1}$, $\cdots$, $x^{n-t_{k+1}-1}A_{k+1}$, $A_{k+2}$, $xA_{k+2}$, $\cdots$, $x^{n-t_{k+2}-1}A_{k+2}$, $\cdots$, $A_{2k}$, $xA_{2k}$, $\cdots$,\linebreak $x^{n-t_{2k}-1}A_{2k}\}$. To show $B$ spans $B'$ we write the set $B'$ as $B'=B_1\cup B_2$, where $B_1=\{A_1, xA_1$, $\cdots$, $x^{n-t_1-1}A_1$, $A_2, xA_2$, $\cdots$, $x^{n-t_2-1}A_2, A_3, xA_3$, $\cdots$, $x^{n-t_3-1}A_3$, $\cdots$, $A_k, xA_k, x^{n-t_k-1}A_k\}$ and $B_2=\{A_{k+1}$, $xA_{k+1}$, $\cdots$, $x^{n-t_{k+1}-1}A_{k+1}$, $A_{k+2}$, $xA_{k+2}$, $\cdots$, $x^{n-t_{k+2}-1}$, $A_{k+3}$, $xA_{k+3}$, $\cdots$, $x^{n-t_{k+3}-1}A_{k+3}$, $\cdots$, $A_{2k}$, $xA_{2k}$, $\cdots$, $x^{n-t_{2k}-1}A_{2k}\}$. First we show that $B$ spans $B_2$ and then we show that $B$ spans $B_1$. To show $B$ spans $B_2$ we divide the 
proof in three cases.
\begin{case} \label{case-1}
Let $t_{k+i}<t_{i+1}$, that is, $t'_i=t_{k+i}$, for $1 \leq i \leq k-1$. We define the set $B_2-B$ as an ordered set :\\
$\{x^{t_{2k-1}-t_{2k}}A_{2k}$, $x^{t_{2k-1}-t_{2k}+1}A_{2k}$, $\cdots$, $x^{t_{2k-2}-t_{2k}-1}A_{2k}$,
$x^{t_{2k-2}-t_{2k-1}}A_{2k-1}$,\linebreak $x^{t_{2k-2}-t_{2k-1}+1}A_{2k-1}$, $\cdots$, $x^{t_{2k-3}-t_{2k-1}-1}A_{2k-1}$,
$x^{t_{2k-2}-t_{2k}}A_{2k}$, $x^{t_{2k-2}-t_{2k}+1}A_{2k}$, $\cdots$, $x^{t_{2k-3}-t_{2k}-1}A_{2k}$,
$x^{t_{2k-3}-t_{2k-2}}A_{2k-2}$, $x^{t_{2k-3}-t_{2k-2}+1}A_{2k-2}$, $\cdots$, $x^{t_{2k-4}-t_{2k-2}-1}A_{2k-2}$,
$x^{t_{2k-3}-t_{2k-1}}A_{2k-1}$, $x^{t_{2k-3}-t_{2k-1}+1}A_{2k-1}$, $\cdots$, $x^{t_{2k-4}-t_{2k-1}-1}A_{2k-1}$,
$x^{t_{2k-3}-t_{2k}}A_{2k}$, $x^{t_{2k-3}-t_{2k}+1}A_{2k}$, $\cdots$, $x^{t_{2k-4}-t_{2k}-1}A_{2k}$,
$\cdots$
$x^{t_{k+1}-t_{k+2}}A_{k+2}$, $x^{t_{k+1}-t_{k+2}+1}A_{k+2}$, $\cdots$, $x^{t_1-t_{k+2}-1}A_{k+2}$,
$x^{t_{k+1}-t_{k+3}}A_{k+3}$, $x^{t_{k+1}-t_{k+3}+1}A_{k+3}$, $\cdots$, $x^{t_1-t_{k+3}-1}A_{k+3}$,
$\cdots$\linebreak
$x^{t_{k+1}-t_{k+i}}A_{k+i}$, $x^{t_{k+1}-t_{k+i}+1}A_{k+i}$, $\cdots$, $x^{t_1-t_{k+i}-1}A_{k+i}$,
$\cdots$
$x^{t_{k+1}-t_{2k-1}}A_{2k-1}$,\linebreak $x^{t_{k+1}-t_{2k-1}+1}A_{2k-1}$, $\cdots$, $x^{t_1-t_{2k-1}-1}A_{2k-1}$,
$x^{t_{k+1}-t_{2k}}A_{2k}$, $x^{t_{k+1}-t_{2k}+1}A_{2k}$, $\cdots$,\linebreak $x^{t_1-t_{2k}-1}A_{2k}$,
$x^{t_1-t_{k+1}}A_{k+1}$, $x^{t_1-t_{k+1}+1}A_{k+1}$, $\cdots$, $x^{n-t_{k+1}-1}A_{k+1}$,
$x^{t_1-t_{k+2}}A_{k+2}$, $x^{t_1-t_{k+2}+1}A_{k+2}$, $\cdots$, $x^{n-t_{k+2}-1}A_{k+2}$,
$\cdots$
$x^{t_1-t_{k+i}}A_{k+i}$, $x^{t_1-t_{k+i}+1}A_{k+i}$, $\cdots$,\linebreak $x^{n-t_{k+i}-1}A_{k+i}$,
$\cdots$
$x^{t_1-t_{2k-1}}A_{2k-1}$, $x^{t_1-t_{2k-1}+1}A_{2k-1}$, $\cdots$, $x^{n-t_{2k-1}-1}A_{2k-1}$,\linebreak
$x^{t_1-t_{2k}}A_{2k}$, $x^{t_1-t_{2k}+1}A_{2k}$, $\cdots$, $x^{n-t_{2k}-1}A_{2k}\}$,
where $x^{t_{2k-1}-t_{2k}}A_{2k}$ is the first and $x^{n-t_{2k}-1}A_{2k}$ is the last element of the set $B_2-B$. Rest of the elements are in the order as they appear in the set $B_2-B$ as given above. To understand the pattern of order of the ordered set $B_2-B$, we write the set as:\\
\begin{align*}
&\{x^{t_{2k-1}-t_{2k}}A_{2k}, ~ x^{t_{2k-1}-t_{2k}+1}A_{2k}, ~ \cdots, ~ x^{t_{2k-2}-t_{2k}-1}A_{2k}\}
\cup\\
&\{x^{t_{2k-2}-t_{2k-1}}A_{2k-1}, ~ x^{t_{2k-2}-t_{2k-1}+1}A_{2k-1}, ~ \cdots, ~ x^{t_{2k-3}-t_{2k-1}-1}A_{2k-1},\\ &x^{t_{2k-2}-t_{2k}}A_{2k}, ~ x^{t_{2k-2}-t_{2k}+1}A_{2k}, ~ \cdots, ~ x^{t_{2k-3}-t_{2k}-1}A_{2k}\}
\cup\\
&\{x^{t_{2k-3}-t_{2k-2}}A_{2k-2}, ~ x^{t_{2k-3}-t_{2k-2}+1}A_{2k-2}, ~ \cdots, ~ x^{t_{2k-4}-t_{2k-2}-1}A_{2k-2},\\
&x^{t_{2k-3}-t_{2k-1}}A_{2k-1}, ~ x^{t_{2k-3}-t_{2k-1}+1}A_{2k-1}, ~ \cdots, ~ x^{t_{2k-4}-t_{2k-1}-1}A_{2k-1},\\
&x^{t_{2k-3}-t_{2k}}A_{2k}, ~ x^{t_{2k-3}-t_{2k}+1}A_{2k}, ~ \cdots, ~ x^{t_{2k-4}-t_{2k}-1}A_{2k}\}\\
&\cup\\
&\vdots\\
&\cup\\
&\{x^{t_{k+1}-t_{k+2}}A_{k+2}, ~ x^{t_{k+1}-t_{k+2}+1}A_{k+2}, ~ \cdots, ~ x^{t_1-t_{k+2}-1}A_{k+2},\\
&x^{t_{k+1}-t_{k+3}}A_{k+3}, ~ x^{t_{k+1}-t_{k+3}+1}A_{k+3}, ~ \cdots, ~ x^{t_1-t_{k+3}-1}A_{k+3},\\
&\vdots\\
&x^{t_{k+1}-t_{k+i}}A_{k+i}, ~ x^{t_{k+1}-t_{k+i}+1}A_{k+i}, ~ \cdots, ~ x^{t_1-t_{k+i}-1}A_{k+i},\\
&\vdots\\
&x^{t_{k+1}-t_{2k-1}}A_{2k-1}, ~ x^{t_{k+1}-t_{2k-1}+1}A_{2k-1}, ~ \cdots, ~ x^{t_1-t_{2k-1}-1}A_{2k-1},\\
&x^{t_{k+1}-t_{2k}}A_{2k}, ~ x^{t_{k+1}-t_{2k}+1}A_{2k}, ~ \cdots, ~ x^{t_1-t_{2k}-1}A_{2k}\}
\cup\\
&\{x^{t_1-t_{k+1}}A_{k+1}, ~ x^{t_1-t_{k+1}+1}A_{k+1}, ~ \cdots, ~ x^{n-t_{k+1}-1}A_{k+1},\\
&x^{t_1-t_{k+2}}A_{k+2}, ~ x^{t_1-t_{k+2}+1}A_{k+2}, ~ \cdots, ~ x^{n-t_{k+2}-1}A_{k+2},\\
&\vdots\\
&x^{t_1-t_{k+i}}A_{k+i}, ~ x^{t_1-t_{k+i}+1}A_{k+i}, ~ \cdots, ~ x^{n-t_{k+i}-1}A_{k+i},\\
&\vdots\\
&x^{t_1-t_{2k-1}}A_{2k-1}, ~ x^{t_1-t_{2k-1}+1}A_{2k-1}, ~ \cdots, ~ x^{n-t_{2k-1}-1}A_{2k-1},\\
&x^{t_1-t_{2k}}A_{2k}, ~ x^{t_1-t_{2k}+1}A_{2k}, ~ \cdots, ~ x^{n-t_{2k}-1}A_{2k}\}\\
\end{align*}
First we show that the first element $x^{t_{2k-1}-t_{2k}}A_{2k} \in B_2-B$ is a linear combination of some elements of $B$ and then we show other elements of the set $B_2-B$ are linear combination of elements of $B$ and its previous elements of the ordered set $B_2-B$. For $i=k$, from Lemma \ref{lm-t_k+i-1-t_k+i}, we have
\begin{equation} \label{eq-t_2k-1-t_2k-a}
x^{t_{2k-1}-t_{2k}}A_{2k}=c_{2k-1}uA_{2k-1}+q_0(x)A_{2k},
\end{equation}
where $\text{deg}(q_0(x))<t_{2k-1}-t_{2k}$. Let
\begin{equation}
q_0(x)=q_{00}+q_{01}x +\cdots+q_{0(t_{2k-1}-t_{2k}-1)}x^{t_{2k-1}-t_{2k}-1},\notag
\end{equation}
where $q_{0i} \in \F_p$. Thus, we have
\begin{multline} \label{eq-t_2k-1-t_2k-b}
x^{t_{2k-1}-t_{2k}}A_{2k}=c_{2k-1}uA_{2k-1}+q_{00}A_{2k}+q_{01}xA_{2k} +\cdots \\+q_{0(t_{2k-1}-t_{2k}-1)}x^{t_{2k-1}-t_{2k}-1}A_{2k}.
\end{multline}
Therefore, $x^{t_{2k-1}-t_{2k}}A_{2k}$ is a linear combination of some elements of $B$, that is $x^{t_{2k-1}-t_{2k}}A_{2k} \in \text{Span}(B)$. Now, we show that $x^{t_{2k-1}-t_{2k}+1}A_{2k} \in \text{Span}(B)$. Multiplying Equation \ref{eq-t_2k-1-t_2k-b} by $x$, we get
\begin{multline} \label{eq-t_2k-1-t_2k+1-a}
x^{t_{2k-1}-t_{2k}+1}A_{2k}=c_{2k-1}uxA_{2k-1}+q_{00}xA_{2k}+q_{01}x^2A_{2k} +\cdots \\+q_{0(t_{2k-1}-t_{2k}-1)}x^{t_{2k-1}-t_{2k}}A_{2k}.
\end{multline}
If we put the value of $x^{t_{2k-1}-t_{2k}}A_{2k}$ from Equation \ref{eq-t_2k-1-t_2k-b} in Equation \ref{eq-t_2k-1-t_2k+1-a}, we get
\begin{multline} \label{eq-t_2k-1-t_2k+1-b}
x^{t_{2k-1}-t_{2k}+1}A_{2k}=(c_{2k-1}uq_{0(t_{2k-1}-t_{2k}-1)}+c_{2k-1}ux)A_{2k-1}+q_{00}q_{0(t_{2k-1}-t_{2k}-1)}A_{2k}\\+(q_{00}+q_{01}q_{0(t_{2k-1}-t_{2k}-1)})xA_{2k}+(q_{01}+q_{02}q_{0(t_{2k-1}-t_{2k}-1)})x^2A_{2k} +\cdots\\+(q_{0(t_{2k-1}-t_{2k}-2)}+q_{0(t_{2k-1}-t_{2k}-1)}q_{0(t_{2k-1}-t_{2k}-1)})x^{t_{2k-1}-t_{2k}-1}A_{2k}.
\end{multline}
Equation \ref{eq-t_2k-1-t_2k+1-b} can be written as
\begin{equation} \label{eq-t_2k-1-t_2k+1-c}
x^{t_{2k-1}-t_{2k}+1}A_{2k}=(c_{2k-1}uq_{0(t_{2k-1}-t_{2k}-1)}+c_{2k-1}ux)A_{2k-1}+q_0'(x)A_{2k},
\end{equation}
where $q_0'(x)=q_{00}q_{0(t_{2k-1}-t_{2k}-1)}+(q_{00}+q_{01}q_{0(t_{2k-1}-t_{2k}-1)})x+(q_{01}+q_{02}q_{0(t_{2k-1}-t_{2k}-1)})\linebreak x^2+\cdots+(q_{0(t_{2k-1}-t_{2k}-2)}+q_{0(t_{2k-1}-t_{2k}-1)}q_{0(t_{2k-1}-t_{2k}-1)})x^{t_{2k-1}-t_{2k}-1}$ and also\linebreak $\text{deg}(q_0'(x))<t_{2k-1}-t_{2k}$. This implies that $x^{t_{2k-1}-t_{2k}+1}A_{2k}$ is a linear combination of its previous elements in the ordered set $B_2-B$ and some elements of $B$. Hence, $x^{t_{2k-1}-t_{2k}+1}A_{2k} \in \text{Span}(B)$. (Note that multiplying Equation \ref{eq-t_2k-1-t_2k-a} by $x$ we get Equation \ref{eq-t_2k-1-t_2k+1-c}. The degree of $q_0(x)$ in Equation \ref{eq-t_2k-1-t_2k-a} and degree of $q_0'(x)$ in Equation \ref{eq-t_2k-1-t_2k+1-c} both are satisfying $\text{deg}(q_0(x)),\text{deg}(q_0'(x))<t_{2k-1}-t_{2k}$. That is, even after multiplying Equation \ref{eq-t_2k-1-t_2k-a} by $x$, the degree of coefficient polynomial of $A_{2k}$ in Equation \ref{eq-t_2k-1-t_2k+1-c} does not exceed by $t_{2k-1}-t_{2k}-1$. In fact, the 
degree of coefficient polynomial of $A_{2k}$ in Equation \ref{eq-t_2k-1-t_2k+1-c} does not exceed by the degree of coefficient polynomial of $A_{2k}$ in Equation \ref{eq-t_2k-1-t_2k-a}. Only the degree of coefficient polynomial of $A_{2k-1}$ is increased by one in Equation \ref{eq-t_2k-1-t_2k+1-c} from Equation \ref{eq-t_2k-1-t_2k-a}). In a similar way, multiplying Equation \ref{eq-t_2k-1-t_2k+1-c} by $x$ we can show that $x^{t_{2k-1}-t_{2k}+2}A_{2k} \in \text{Span}(B)$. And also it can be shown that degree of coefficient polynomial of $A_{2k}$ will not be increased. Only the degree of coefficient polynomial of $A_{2k-1}$ will be increased by one. In this similar way after $t_{2k-2}-t_{2k-1}-1$ times we will get
\begin{equation} \label{eq-t_2k-2-t_2k-1_A_2k}
x^{t_{2k-2}-t_{2k}-1}A_{2k}=q_{-1}(x)A_{2k-1}+q_0''(x)A_{2k},
\end{equation}
for some $q_{-1}(x)$, $q_0''(x) \in \F_p[x]$ with $\text{deg}(q_{-1}(x))=t_{2k-2}-t_{2k-1}-1$ and $\text{deg}(q_0''(x))<t_{2k-1}-t_{2k}$. From Equation \ref{eq-t_2k-2-t_2k-1_A_2k} we can say that $x^{t_{2k-2}-t_{2k}-1}A_{2k} \in \text{Span}(B)$. Now the next term of $x^{t_{2k-2}-t_{2k}-1}A_{2k}$ in ordered set $B_2-B$ is $x^{t_{2k-2}-t_{2k-1}}A_{2k-1}$. To show $x^{t_{2k-2}-t_{2k-1}}A_{2k-1} \in \text{Span}(B)$ we proceed as follows. For $i=k-1$, from Lemma \ref{lm-t_k+i-1-t_k+i}, we have 
\begin{equation} \label{eq-t_2k-2-t_2k-1}
x^{t_{2k-2}-t_{2k-1}}A_{2k-1}=c_{2k-2}uA_{2k-2}+q_0(x)A_{2k-1}+q_1(x)A_{2k}
\end{equation}
where $\text{deg}(q_0(x))<t_{2k-2}-t_{2k-1}$ and $\text{deg}(q_1(x))<(t_{2k-2}-t_{2k-1}) ~ \text{or} ~ (t_{2k-1}-t_{2k})$. In the above discussion we have shown that $x^iA_{2k} \in \text{Span}(B)$, for $1 \leq i \leq t_{2k-2}-t_{2k}-1$. Now we have $t_{2k-2}>t_{2k-1}>t_{2k}$ which implies that $t_{2k-2}-t_{2k-1}$, $t_{2k-1}-t_{2k}<t_{2k-2}-t_{2k}$. This gives $\text{deg}(q_1(x))<t_{2k-2}-t_{2k}$. Therefore $q_1(x)A_{2k}$ is a linear combination of some element of $B$ and previous elements of term $x^{t_{2k-2}-t_{2k}-1}A_{2k}$ in the ordered set $B_2-B$. That is $q_1(x)A_{2k} \in \text{Span}(B)$. And also both the term $c_{2k-2}uA_{2k-2}$, $q_0(x)A_{2k-1} \in \text{Span}(B)$ (since, $\text{deg}(q_0(x))<t_{2k-2}-t_{2k-1}$). Therefore $x^{t_{2k-2}-t_{2k-1}}A_{2k-1}$ is a linear combination of some elements of $B$ and its previous elements in the ordered set $B_2-B$. Therefore, $x^{t_{2k-2}-t_{2k-1}}A_{2k-1} \in \text{Span}(B)$. Now, we show that $x^{t_{2k-2}-t_{2k-1}+1}A_{2k-1} \in \text{Span}(B)$. We follow 
the same techniques as above. After putting the value of $x^{t_{2k-2}-t_{2k-1}}A_{2k-1}$ and $x^{t_{2k-2}-t_{2k-1}}A_{2k}$ (or $x^{t_{2k-1}-t_{2k}}A_{2k}$) in the equation obtained by multiplying Equation \ref{eq-t_2k-2-t_2k-1} by $x$, we get
\begin{equation} \label{eq-t_2k-2-t_2k-1+1}
x^{t_{2k-2}-t_{2k-1}+1}A_{2k-1}=(q_{00}'c_{2k-2}u+c_{2k-2}ux)A_{2k-2}+q_0'(x)A_{2k-1}+q_1'(x)A_{2k}
\end{equation}
for some $q_{00}' \in \F_p$ and $q_0'(x), q_1'(x) \in \F_p[x]$ such that $\text{deg}(q_0'(x))<t_{2k-2}-t_{2k-1}$ and $\text{deg}(q_1'(x))<(t_{2k-2}-t_{2k-1}) ~ \text{or} ~ (t_{2k-1}-t_{2k})$. (Note that after multiplying Equation \ref{eq-t_2k-2-t_2k-1} by $x$ we get Equation \ref{eq-t_2k-2-t_2k-1+1} and the degree of coefficient polynomial of $A_{2k-1}$ and $A_{2k}$ does not exceed by $t_{2k-2}-t_{2k-1}-1$ and $(t_{2k-2}-t_{2k-1}-1) ~ \text{or} ~ (t_{2k-1}-t_{2k}-1)$ respectively. In fact the degree of coefficient polynomial of $A_{2k-1}$ and $A_{2k}$ in Equation \ref{eq-t_2k-2-t_2k-1+1} are not exceed by the degree of coefficient polynomial of $A_{2k-1}$ and $A_{2k}$ respectively in Equation \ref{eq-t_2k-2-t_2k-1}. Only the degree of coefficient polynomial of $A_{2k-2}$ is increased by one in Equation \ref{eq-t_2k-2-t_2k-1+1}. This fact can be shown by same techniques as shown above for getting Equation \ref{eq-t_2k-1-t_2k+1-c} from Equation \ref{eq-t_2k-1-t_2k-a}). Therefore, it is clear from Equation 
\ref{eq-t_2k-2-t_2k-1+1} that $x^{t_{2k-2}-t_{2k-1}+1}A_{2k-1} \in \text{Span}(B)$. In a similar way, multiplying Equation \ref{eq-t_2k-2-t_2k-1+1} by $x$ we can show that $x^{t_{2k-2}-t_{2k-1}+2}A_{2k-1} \in \text{Span}(B)$. And also it can be shown that degree of coefficient polynomial of $A_{2k-1}$ and $A_{2k}$ will not be increased. Only the degree of coefficient polynomial of $A_{2k-2}$ will be increased by one. In this similar way after $t_{2k-3}-t_{2k-2}-1$ times we will get
\begin{equation} \label{eq-t_2k-3-t_2k-1-1_A_2k-1}
x^{t_{2k-3}-t_{2k-1}-1}A_{2k-1}=q_{-1}(x)A_{2k-2}+q_0''(x)A_{2k-1}+q_1''(x)A_{2k}
\end{equation}
for some $q_{-1}(x)$, $q_0''(x)$ and $q_1''(x) \in \F_p[x]$ with $\text{deg}(q_{-1}(x))=t_{2k-3}-t_{2k-2}-1$ and $\text{deg}(q_0''(x))<t_{2k-2}-t_{2k-1}$ and $\text{deg}(q_1''(x))<(t_{2k-2}-t_{2k-1})~ \text{or} ~(t_{2k-1}-t_{2k})$. From Equation \ref{eq-t_2k-3-t_2k-1-1_A_2k-1} we can say that $x^{t_{2k-3}-t_{2k-1}-1}A_{2k-1} \in \text{Span}(B)$. Now the next term of $x^{t_{2k-3}-t_{2k-1}-1}A_{2k-1}$ in the ordered set $B_2-B$ is $x^{t_{2k-2}-t_{2k}}A_{2k}$. To show $x^{t_{2k-2}-t_{2k}}A_{2k} \in \text{Span}(B)$ we multiply Equation \ref{eq-t_2k-2-t_2k-1_A_2k} by $x$. Then we get $x^{t_{2k-2}-t_{2k}}A_{2k} \in \text{Span}(B)$. Repeating this for $t_{2k-3}-t_{2k-2}$ times we can show that $x^{t_{2k-3}-t_{2k}-1}A_{2k} \in \text{Span}(B)$. In this fashion, for $i=k-2,k-3,\cdots, 2$, from Lemma \ref{lm-t_k+i-1-t_k+i}, we can show that the terms up to $x^{t_1-t_{2k}-1}A_{2k} \in \text{Span}(B)$ one by one. Now we have to show that $x^{t_1-t_{k+1}}A_{k+1} \in \text{Span}(B)$. For $i=1$, by Lemma \ref{lm-t_i-t_k+i}, 
we have the equation
\begin{multline} \label{eq-t_1-t_k+1}
x^{t_1-t_{k+1}}A_{k+1}=c_1vA_1+q_0(x)A_{k+1}+q_1(x)A_{k+2}\\+q_2(x)A_{k+3}+ \cdots +q_{k-1}(x)A_{2k},
\end{multline}
for some $c_1 \in \F_p$ and $q_0(x)$, $q_1(x)$, $\cdots$, $q_{k-1}(x)\in \F_p[x]$ with $\text{deg}(q_0(x))<t_1-t_{k+1}$, $\text{deg}(q_1(x))<(t_1-t_{k+1})$ or $(t_{k+1}-t_{k+2})$, $\text{deg}(q_2(x))<(t_{1}-t_{k+1})$ or $(t_{k+1}-t_{k+2})$ or $(t_{k+2}-t_{k+3})$, $\cdots$ and $\text{deg}(q_{k-1}(x))<(t_{1}-t_{k+1})$ or $(t_{k+1}-t_{k+2})$ or $(t_{k+2}-t_{k+3})$ or $\cdots$ or $(t_{2k-1}-t_{2k})$. We have shown up to now that $x^{t_1-t_{k+2}-1}A_{k+2}$, $x^{t_1-t_{k+3}-1}A_{k+3}$, $\cdots$, $x^{t_1-t_{2k}-1}A_{2k} \in \text{Span}(B)$. Therefore, the terms $q_{i-1}(x)A_{k+i}$, for $1\leq i\leq k$ in Equation \ref{eq-t_1-t_k+1} are linear combination of some previous elements of $x^{t_1-t_{k+i}}A_{k+i}$, for $2\leq i\leq k$ in the ordered set $B_2-B$ and some elements of the set  $B$. This implies that, $q_{i-1}(x)A_{k+i} \in \text{Span}(B)$, for $1\leq i\leq k$. Therefore, $x^{t_1-t_{k+1}}A_{k+1} \in \text{Span}(B)$. Multiplying Equation \ref{eq-t_1-t_k+1} by $x$, in a similar fashion as above we can show that 
$x^{t_1-t_{k+1}+1}A_{k+1} \in \text{Span}(B)$. Repeating this for $n-t_1$ times we can show that $x^{n-t_{k+1}+1}A_{k+1} \in \text{Span}(B)$. Again, to show the elements from $x^{t_1-t_{k+2}}A_{k+2}$ to $x^{n-t_{2k}-1}A_{2k}$ of the ordered set $B_2-B$ are in $\text{Span}(B)$, we use Lemma \ref{lm-t_k+i-1-t_k+i} and apply same techniques as applied above.
\end{case}

\begin{case} \label{case-2}
Let $t_{k+i}>t_{i+1}$, that is, $t'_i=t_{i+1}$, for $1 \leq i \leq k-1$. We define the set $B_2-B$ as an ordered set :\\
$\{x^{t_k-t_{2k}}A_{2k}$, $x^{t_k-t_{2k}+1}A_{2k}$, $\cdots$, $x^{n-t_{2k}-1}A_{2k}$,
$x^{t_{k-1}-t_{2k-1}}A_{2k-1}$, $x^{t_{k-1}-t_{2k-1}+1}A_{2k-1}$, $\cdots$, $x^{n-t_{2k-1}-1}A_{2k-1}$,
$\cdots$,
$x^{t_i-t_{k+i}}A_{k+i}$, $x^{t_i-t_{k+i}+1}A_{k+i}$, $\cdots$, $x^{n-t_{k+i}-1}A_{k+i}$,
$\cdots$,
$x^{t_2-t_{k+2}}A_{k+2}$, $x^{t_2-t_{k+2}+1}A_{k+2}$, $\cdots$, $x^{n-t_{k+2}-1}A_{k+2}$,
$x^{t_1-t_{k+1}}A_{k+1}$, $x^{t_1-t_{k+1}+1}A_{k+1}$, $\cdots$, $x^{n-t_{k+1}-1}A_{k+1}\}$,\\
where $x^{t_k-t_{2k}}A_{2k}$ is the first and $x^{n-t_{k+1}-1}A_{k+1}$ is the last element of the set $B_2-B$. Rest of the elements are in the order as they appear in the set $B_2-B$ given above.
For $i=k$, from Lemma \ref{lm-t_i-t_k+i}, we have 
\begin{equation} \label{eq-t_k-t_2k-a}
x^{t_k-t_{2k}}A_{2k}=c_kvA_k+q_0(x)A_{2k},
\end{equation}
where $\text{deg}(q_0(x))<t_k-t_{2k}$. Let
\begin{equation}
q_0(x)=q_{00}+q_{01}x +\cdots +q_{0(t_k-t_{2k}-1)}x^{t_k-t_{2k}-1},\notag
\end{equation}
where $q_{0i} \in \F_p$. Thus we have
\begin{equation} \label{eq-t_k-t_2k-b}
x^{t_k-t_{2k}}A_{2k}=c_kvA_k+q_{00}A_{2k}+q_{01}xA_{2k} +\cdots +q_{0(t_k-t_{2k}-1)}x^{t_k-t_{2k}-1}A_{2k}.
\end{equation}
Therefore, $x^{t_k-t_{2k}}A_{2k}$ is a linear combination of some elements of $B$, that is $x^{t_k-t_{2k}}A_{2k} \in \text{Span}(B)$. Next, we show that $x^{t_k-t_{2k}+1}A_{2k} \in \text{Span}(B)$. Multiplying Equation \ref{eq-t_k-t_2k-b} by $x$ we get
\begin{equation} \label{eq-t_k-t_2k+1-a}
x^{t_k-t_{2k}+1}A_{2k}=c_kvxA_k+q_{00}xA_{2k}+q_{01}x^2A_{2k} +\cdots +q_{0(t_k-t_{2k}-1)}x^{t_k-t_{2k}}A_{2k}.
\end{equation}
As in Case \ref{case-1}, if we put the value of $x^{t_k-t_{2k}}A_{2k}$ from Equation \ref{eq-t_k-t_2k-b} in Equation \ref{eq-t_k-t_2k+1-a}, we get
\begin{equation} \label{eq-t_k-t_2k+1-b}
x^{t_k-t_{2k}+1}A_{2k}=(c_kq_{0(t_k-t_{2k}-1)}q_{00}'v+c_kvx)A_k+q_0'(x)A_{2k}.
\end{equation}
for some $q_0'(x) \in \F_p[x]$ such that $\text{deg}(q_0'(x))<t_k-t_{2k}$. (Note that multiplying Equation \ref{eq-t_k-t_2k-a} by $x$ we get Equation \ref{eq-t_k-t_2k+1-b} and the degree of the coefficient polynomial of $A_{2k}$ is not exceed by $t_k-t_{2k}-1$. In fact the degree of coefficient polynomial of $A_{2k}$ in Equation \ref{eq-t_k-t_2k+1-b} does not exceed by the degree of coefficient polynomial of $A_{2k}$ in Equation \ref{eq-t_k-t_2k-a}. Only the degree of coefficient polynomial of $A_k$ is increased by one. Proof of this fact can be shown by same technique as shown in Case \ref{case-1} for getting Equation \ref{eq-t_2k-1-t_2k+1-c} from Equation \ref{eq-t_2k-1-t_2k-a}). From Lemma \ref{lm-vA_i}, we have $x^ivA_k \in \text{Span}(B)$ for $0 \leq i \leq n-t_k-1$. This implies that $x^{t_k-t_{2k}+1}A_{2k}$ is a linear combination of its previous elements in the ordered set $B_2-B$ and some elements of $B$. Hence, $x^{t_k-t_{2k}+1}A_{2k} \in \text{Span}(B)$. In a similar way, multiplying Equation
\ref{eq-t_k-t_2k+1-b} by $x$ we can show that $x^{t_k-t_{2k}+2}A_{2k} \in \text{Span}(B)$. And also it can be shown that degree of coefficient polynomial of $A_{2k}$ will not be increased. Only the degree of coefficient polynomial of $A_k$ will be increased by one. In this similar way after $n-t_k$ times we can show that $x^{n-t_{2k}-1}A_{2k} \in \text{Span}(B)$. Now we show that $x^{t_{k-1}-t_{2k-1}}A_{2k-1} \in \text{Span}(B)$. For $i=k-1$, from Lemma \ref{lm-t_i-t_k+i}, we have
\begin{equation} \label{eq-t_k-1-t_2k-1}
x^{t_{k-1}-t_{2k-1}}A_{2k-1}=c_{k-1}vA_{k-1}+q_0(x)A_{2k-1}+q_1(x)A_{2k},
\end{equation}
where $\text{deg}(q_0(x))<t_{k-1}-t_{2k-1}$ and $\text{deg}(q_1(x))<(t_{k-1}-t_{2k-1}) ~ \text{or} ~ (t_{2k-1}-t_{2k})$. In above discussion (in Case \ref{case-2}), we have shown that $x^iA_{2k} \in \text{Span}(B)$ for $0 \leq i \leq n-t_{2k}-1$. Therefore, $q_1(x)A_{2k} \in \text{Span}(B)$. Also, from Lemma \ref{lm-vA_i}, we have $x^ivA_{k-1}\in \text{Span}(B)$ for $0 \leq i \leq n-t_{k-1}-1$.  And, we have $q_0(x)A_{2k-1} \in \text{Span}(B)$ (since $\text{deg}(q_0(x))<t_{k-1}-t_{2k-1}$). Therefore, $x^{t_{k-1}-t_{2k-1}}A_{2k-1} \in \text{Span}(B)$. Now we show that $x^{t_{k-1}-t_{2k-1}+1}A_{2k-1} \in \text{Span}(B)$. We follow the same technique as shown above. After putting the value of $x^{t_{k-1}-t_{2k-1}}A_{2k-1}$ and $x^{t_{k-1}-t_{2k-1}}A_{2k}$ (or $x^{t_{2k-1}-t_{2k}}A_{2k}$) in the equation obtained by Equation \ref{eq-t_k-1-t_2k-1} by multiplying $x$, we get
\begin{equation} \label{eq-t_k-1-t_2k-1+1}
x^{t_{k-1}-t_{2k-1}+1}A_{2k-1}=(q_{00}'c_{k-1}v+c_{k-1}vx)A_{k-1}+q_0'(x)A_{2k-1}+q_1'(x)A_{2k}
\end{equation}
for some $q_{00}' \in \F_p$ and $q_0'(x), q_1'(x) \in \F_p[x]$ such that $\text{deg}(q_0'(x))<t_{k-1}-t_{2k-1}$ and $\text{deg}(q_1'(x))<(t_{k-1}-t_{2k-1}) ~ \text{or} ~ (t_{2k-1}-t_{2k})$. (Note that after multiplying Equation \ref{eq-t_k-1-t_2k-1} by $x$ we get Equation \ref{eq-t_k-1-t_2k-1+1} and the degree of coefficient polynomial of $A_{2k-1}$ and $A_{2k}$ are not exceed by $t_{k-1}-t_{2k-1}-1$ and $t_{k-1}-t_{2k-1}-1 ~ \text{or} ~ t_{2k-1}-t_{2k}-1$. In fact the degree of coefficient polynomial of $A_{2k-1}$ and $A_{2k}$ in Equation \ref{eq-t_k-1-t_2k-1+1} are not exceed by the degree of coefficient polynomial of $A_{2k-1}$ and $A_{2k}$ in Equation \ref{eq-t_k-1-t_2k-1}. Only the degree of coefficient polynomial $A_{k-1}$ is increased by one. This fact can be shown by same technique as shown in Case \ref{case-1} for getting Equation \ref{eq-t_2k-1-t_2k+1-c} from Equation \ref{eq-t_2k-1-t_2k-a}). Therefore, it is clear from Equation \ref{eq-t_k-1-t_2k-1+1} that $x^{t_{2k-2}-t_{2k-1}+1}A_{2k-1} \in \
text{Span}(B)$. In a similar way, by multiplying Equation \eqref{eq-t_k-1-t_2k-1+1} by $x$, we can show that $x^{t_{k-1}-t_{2k-1}+2}A_{2k-1} \in \text{Span}(B)$. And also it can be shown that degree of coefficient polynomial of $A_{2k-1}$ and $A_{2k}$ will not be increased. Only the degree of coefficient polynomial of $A_{k-1}$ will be increased by one. In this similar way after $n-t_{k-1}$ times we will get $x^{n-t_{2k-1}-1}A_{2k-1} \in \text{Span}(B)$. In this fashion, for $i=k-2,k-3,\cdots, 1$, from Lemma \ref{lm-t_i-t_k+i}, we can show the rest of the terms in the ordered set $B_2-B$ are belongs to $\text{Span}(B)$ that is $x^{n-t_{k+1}-1}A_{k+1} \in \text{Span}(B)$.
\end{case}

\begin{case} \label{case-3}
Let $I=\{l_1,l_2, \cdots, l_r\}$ for $1 \leq r \leq k-1$ and $1\leq l_1<l_2<\cdots<l_{r-1}<l_{r}\leq k-1$. Let $t_{k+i}<t_{i+1}$ for  $i\notin I$, $1 \leq i \leq k-1$ and $t_{k+i}>t_{i+1}$ for  $i\in I$. If $r=k-1$ that is $l_r=k-1$ then the this case will be reduced to Case \ref{case-2} and if the set $I$ is empty then this case will be reduced to Case \ref{case-1}. We define the set $B_2-B$ as an ordered set :\\
$\{x^{t_{2k-1}-t_{2k}}A_{2k}$, $x^{t_{2k-1}-t_{2k}+1}A_{2k}$, $\cdots$, $x^{t_{2k-2}-t_{2k}-1}A_{2k}\}$
$\cup$\\
$\{x^{t_{2k-2}-t_{2k-1}}A_{2k-1}$, $x^{t_{2k-2}-t_{2k-1}+1}A_{2k-1}$, $\cdots$, $x^{t_{2k-3}-t_{2k-1}-1}A_{2k-1}$,\\
$x^{t_{2k-2}-t_{2k}}A_{2k}$, $x^{t_{2k-2}-t_{2k}+1}A_{2k}$, $\cdots$, $x^{t_{2k-3}-t_{2k}-1}A_{2k}\}$
$\cup$\\
$\{x^{t_{2k-3}-t_{2k-2}}A_{2k-2}$, $x^{t_{2k-3}-t_{2k-2}+1}A_{2k-2}$, $\cdots$, $x^{t_{2k-4}-t_{2k-2}-1}A_{2k-2}$,\\
$x^{t_{2k-3}-t_{2k-1}}A_{2k-1}$, $x^{t_{2k-3}-t_{2k-1}+1}A_{2k-1}$, $\cdots$, $x^{t_{2k-4}-t_{2k-1}-1}A_{2k-1}$,\\
$x^{t_{2k-3}-t_{2k}}A_{2k}$, $x^{t_{2k-3}-t_{2k}+1}A_{2k}$, $\cdots$, $x^{t_{2k-4}-t_{2k}-1}A_{2k}\}$\\
$\cup$\\
$\vdots$\\
$\cup$\\
$\{x^{t_{k+l_r+1}-t_{k+l_r+2}}A_{k+l_r+2}$, $x^{t_{k+l_r+1}-t_{k+l_r+2}+1}A_{k+l_r+2}$, $\cdots$, $x^{t_{l_r+1}-t_{k+l_r+2}-1}A_{k+l_r+2}$,\\
$x^{t_{k+l_r+1}-t_{k+l_r+3}}A_{k+l_r+3}$, $x^{t_{k+l_r+1}-t_{k+l_r+3}+1}A_{k+l_r+3}$, $\cdots$, $x^{t_{l_r+1}-t_{k+l_r+3}-1}A_{k+l_r+3}$,\\
$\vdots$\\
$x^{t_{k+l_r+1}-t_{k+l_r+i}}A_{k+l_r+i}$, $x^{t_{k+l_r+1}-t_{k+l_r+i}+1}A_{k+l_r+i}$, $\cdots$, $x^{t_{l_r+1}-t_{k+l_r+i}-1}A_{k+l_r+i}$,\\
$\vdots$\\
$x^{t_{k+l_r+1}-t_{2k-1}}A_{2k-1}$, $x^{t_{k+l_r+1}-t_{2k-1}+1}A_{2k-1}$, $\cdots$, $x^{t_{l_r+1}-t_{2k-1}-1}A_{2k-1}$,\\
$x^{t_{k+l_r+1}-t_{2k}}A_{2k}$, $x^{t_{k+l_r+1}-t_{2k}+1}A_{2k}$, $\cdots$, $x^{t_{l_r+1}-t_{2k}-1}A_{2k}\}$
$\cup$\\
$\{x^{t_{l_r+1}-t_{k+l_r+1}}A_{k+l_r+1}$, $x^{t_{l_r+1}-t_{k+l_r+1}+1}A_{k+l_r+1}$, $\cdots$, $x^{n-t_{k+l_r+1}-1}A_{k+l_r+1}$,\\
$x^{t_{l_r+1}-t_{k+l_r+2}}A_{k+l_r+2}$, $x^{t_{l_r+1}-t_{k+l_r+2}+1}A_{k+l_r+2}$, $\cdots$, $x^{n-t_{k+l_r+2}-1}A_{k+l_r+2}$,\\
$\vdots$\\
$x^{t_{l_r+1}-t_{k+l_r+i}}A_{k+l_r+i}$, $x^{t_{l_r+1}-t_{k+l_r+i}+1}A_{k+l_r+i}$, $\cdots$, $x^{n-t_{k+l_r+i}-1}A_{k+l_r+i}$,\\
$\vdots$\\
$x^{t_{l_r+1}-t_{2k-1}}A_{2k-1}$, $x^{t_{l_r+1}-t_{2k-1}+1}A_{2k-1}$, $\cdots$, $x^{n-t_{2k-1}-1}A_{2k-1}$,\\
$x^{t_{l_r+1}-t_{2k}}A_{2k}$, $x^{t_{l_r+1}-t_{2k}+1}A_{2k}$, $\cdots$, $x^{n-t_{2k}-1}A_{2k}\}$\\
$\cup$\\
$\{x^{t_{k+l_j-1}-t_{k+l_j}}A_{k+l_j}$, $x^{t_{k+l_j-1}-t_{k+l_j}+1}A_{k+l_j}$, $\cdots$, $x^{t_{k+l_j-2}-t_{k+l_j}-1}A_{k+l_j}\}$
$\cup$\\
$\{x^{t_{k+l_j-2}-t_{k+l_j-1}}A_{k+l_j-1}$, $x^{t_{k+l_j-2}-t_{k+l_j-1}+1}A_{k+l_j-1}$, $\cdots$, $x^{t_{k+l_j-3}-t_{k+l_j-1}-1}A_{k+l_j-1}$,\\
$x^{t_{k+l_j-2}-t_{k+l_j}}A_{k+l_j}$, $x^{t_{k+l_j-2}-t_{k+l_j}+1}A_{k+l_j}$, $\cdots$, $x^{t_{k+l_j-3}-t_{k+l_j}-1}A_{k+l_j}\}$
$\cup$\\
$\{x^{t_{k+l_j-3}-t_{k+l_j-2}}A_{k+l_j-2}$, $x^{t_{k+l_j-3}-t_{k+l_j-2}+1}A_{k+l_j-2}$, $\cdots$, $x^{t_{k+l_j-4}-t_{k+l_j-2}-1}A_{k+l_j-2}$,\\
$x^{t_{k+l_j-3}-t_{k+l_j-1}}A_{k+l_j-1}$, $x^{t_{k+l_j-3}-t_{k+l_j-1}+1}A_{k+l_j-1}$, $\cdots$, $x^{t_{k+l_j-4}-t_{k+l_j-1}-1}A_{k+l_j-1}$,\\
$x^{t_{k+l_j-3}-t_{k+l_j}}A_{k+l_j}$, $x^{t_{k+l_j-3}-t_{k+l_j}+1}A_{k+l_j}$, $\cdots$, $x^{t_{k+l_j-4}-t_{k+l_j}-1}A_{k+l_j}\}$\\
$\cup$\\
$\vdots$\\
$\cup$\\
$\{x^{t_{k+l_{j-1}+2}-t_{k+l_{j-1}+3}}A_{k+l_{j-1}+3}$, $x^{t_{k+l_{j-1}+2}-t_{k+l_{j-1}+3}+1}A_{k+l_{j-1}+3}$, $\cdots$,\\ $x^{t_{k+l_{j-1}+1}-t_{k+l_{j-1}+3}-1}A_{k+l_{j-1}+3}$,\\
$x^{t_{k+l_{j-1}+2}-t_{k+l_{j-1}+4}}A_{k+l_{j-1}+4}$, $x^{t_{k+l_{j-1}+2}-t_{k+l_{j-1}+4}+1}A_{k+l_{j-1}+4}$, $\cdots$,\\ $x^{t_{k+l_{j-1}+1}-t_{k+l_{j-1}+4}-1}A_{k+l_{j-1}+4}$,\\
$\vdots$\\
$x^{t_{k+l_{j-1}+2}-t_{k+l_{j-1}+i}}A_{k+l_{j-1}+i}$, $x^{t_{k+l_{j-1}+2}-t_{k+l_{j-1}+i}+1}A_{k+l_{j-1}+i}$, $\cdots$,\\ $x^{t_{k+l_{j-1}+1}-t_{k+l_{j-1}+i}-1}A_{k+l_{j-1}+i}$,\\
$\vdots$\\
$x^{t_{k+l_{j-1}+2}-t_{k+l_j-1}}A_{k+l_j-1}$, $x^{t_{k+l_{j-1}+2}-t_{k+l_j-1}+1}A_{k+l_j-1}$, $\cdots$,\\ $x^{t_{k+l_{j-1}+1}-t_{k+l_j-1}-1}A_{k+l_j-1}$,\\
$x^{t_{k+l_{j-1}+2}-t_{k+l_j}}A_{k+l_j}$, $x^{t_{k+l_{j-1}+2}-t_{k+l_j}+1}A_{k+l_j}$, $\cdots$,\\ $x^{t_{k+l_{j-1}+1}-t_{k+l_j}-1}A_{k+l_j}\}$
$\cup$\\
$\{x^{t_{k+l_{j-1}+1}-t_{k+l_{j-1}+2}}A_{k+l_{j-1}+2}$, $x^{t_{k+l_{j-1}+1}-t_{k+l_{j-1}+2}+1}A_{k+l_{j-1}+2}$, $\cdots$,\\ $x^{t_{l_{j-1}+1}-t_{k+l_{j-1}+2}-1}A_{k+l_{j-1}+2}$,\\
$x^{t_{k+l_{j-1}+1}-t_{k+l_{j-1}+3}}A_{k+l_{j-1}+3}$, $x^{t_{k+l_{j-1}+1}-t_{k+l_{j-1}+3}+1}A_{k+l_{j-1}+3}$, $\cdots$,\\ $x^{t_{l_{j-1}+1}-t_{k+l_{j-1}+3}-1}A_{k+l_{j-1}+3}$,\\
$\vdots$\\
$x^{t_{k+l_{j-1}+1}-t_{k+l_{j-1}+i}}A_{k+l_{j-1}+i}$, $x^{t_{k+l_{j-1}+1}-t_{k+l_{j-1}+i}+1}A_{k+l_{j-1}+i}$, $\cdots$,\\ $x^{t_{l_{j-1}+1}-t_{k+l_{j-1}+i}-1}A_{k+l_{j-1}+i}$,\\
$\vdots$\\
$x^{t_{k+l_{j-1}+1}-t_{k+l_j-1}}A_{k+l_j-1}$, $x^{t_{k+l_{j-1}+1}-t_{k+l_j-1}+1}A_{k+l_j-1}$, $\cdots$,\\ $x^{t_{l_{j-1}+1}-t_{k+l_j-1}-1}A_{k+l_j-1}$,\\
$x^{t_{k+l_{j-1}+1}-t_{k+l_j}}A_{k+l_j}$, $x^{t_{k+l_{j-1}+1}-t_{k+l_j}+1}A_{k+l_j}$, $\cdots$,\\ $x^{t_{l_{j-1}+1}-t_{k+l_j}-1}A_{k+l_j}\}$
$\cup$\\
$\{x^{t_{l_{j-1}+1}-t_{k+l_{j-1}+1}}A_{k+l_{j-1}+1}$, $x^{t_{l_{j-1}+1}-t_{k+l_{j-1}+1}+1}A_{k+l_{j-1}+1}$, $\cdots$,\\ $x^{n-t_{k+l_{j-1}+1}-1}A_{k+l_{j-1}+1}$,\\
$x^{t_{l_{j-1}+1}-t_{k+l_{j-1}+2}}A_{k+l_{j-1}+2}$, $x^{t_{l_{j-1}+1}-t_{k+l_{j-1}+2}+1}A_{k+l_{j-1}+2}$, $\cdots$,\\ $x^{n-t_{k+l_{j-1}+2}-1}A_{k+l_{j-1}+2}$,\\
$\vdots$\\
$x^{t_{l_{j-1}+1}-t_{k+l_{j-1}+i}}A_{k+l_{j-1}+i}$, $x^{t_{l_{j-1}+1}-t_{k+l_{j-1}+i}+1}A_{k+l_{j-1}+i}$, $\cdots$,\\ $x^{n-t_{k+l_{j-1}+i}-1}A_{k+l_{j-1}+i}$,\\
$\vdots$\\
$x^{t_{l_{j-1}+1}-t_{k+l_j-1}}A_{k+l_j-1}$, $x^{t_{l_{j-1}+1}-t_{k+l_j-1}+1}A_{k+l_j-1}$, $\cdots$,\\ $x^{n-t_{k+l_j-1}-1}A_{k+l_j-1}$,\\
$x^{t_{l_{j-1}+1}-t_{k+l_j}}A_{k+l_j}$, $x^{t_{l_{j-1}+1}-t_{k+l_j}+1}A_{k+l_j}$, $\cdots$,\\ $x^{n-t_{k+l_j}-1}A_{k+l_j}\}$\\
$\cup$\\
$\{x^{t_{k+l_1-1}-t_{k+l_1}}A_{k+l_1}$, $x^{t_{k+l_1-1}-t_{k+l_1}+1}A_{k+l_1}$, $\cdots$, $x^{t_{k+l_1-2}-t_{k+l_1}-1}A_{k+l_1}\}$
$\cup$\\
$\{x^{t_{k+l_1-2}-t_{k+l_1-1}}A_{k+l_1-1}$, $x^{t_{k+l_1-2}-t_{k+l_1-1}+1}A_{k+l_1-1}$, $\cdots$, $x^{t_{k+l_1-3}-t_{k+l_1-1}-1}A_{k+l_1-1}$,\\
$x^{t_{k+l_1-2}-t_{k+l_1}}A_{k+l_1}$, $x^{t_{k+l_1-2}-t_{k+l_1}+1}A_{k+l_1}$, $\cdots$, $x^{t_{k+l_1-3}-t_{k+l_1}-1}A_{k+l_1}\}$
$\cup$\\
$\{x^{t_{k+l_1-3}-t_{k+l_1-2}}A_{k+l_1-2}$, $x^{t_{k+l_1-3}-t_{k+l_1-2}+1}A_{k+l_1-2}$, $\cdots$, $x^{t_{k+l_1-4}-t_{k+l_1-2}-1}A_{k+l_1-2}$,\\
$x^{t_{k+l_1-3}-t_{k+l_1-1}}A_{k+l_1-1}$, $x^{t_{k+l_1-3}-t_{k+l_1-1}+1}A_{k+l_1-1}$, $\cdots$, $x^{t_{k+l_1-4}-t_{k+l_1-1}-1}A_{k+l_1-1}$,\\
$x^{t_{k+l_1-3}-t_{k+l_1}}A_{k+l_1}$, $x^{t_{k+l_1-3}-t_{k+l_1}+1}A_{k+l_1}$, $\cdots$, $x^{t_{k+l_1-4}-t_{k+l_1}-1}A_{k+l_1}\}$\\
$\cup$\\
$\vdots$\\
$\cup$\\
$\{x^{t_{k+2}-t_{k+3}}A_{k+3}$, $x^{t_{k+2}-t_{k+3}+1}A_{k+3}$, $\cdots$, $x^{t_{k+1}-t_{k+3}-1}A_{k+3}$,\\
$x^{t_{k+2}-t_{k+4}}A_{k+4}$, $x^{t_{k+2}-t_{k+4}+1}A_{k+4}$, $\cdots$, $x^{t_{k+1}-t_{k+4}-1}A_{k+4}$,\\
$\vdots$\\
$x^{t_{k+2}-t_{k+l_1-i}}A_{k+l_1-i}$, $x^{t_{k+2}-t_{k+l_1-i}+1}A_{k+l_1-i}$, $\cdots$, $x^{t_{k+1}-t_{k+l_1-i}-1}A_{k+l_1-i}$,\\
$\vdots$\\
$x^{t_{k+2}-t_{k+l_1-1}}A_{k+l_1-1}$, $x^{t_{k+2}-t_{k+l_1-1}+1}A_{k+l_1-1}$, $\cdots$, $x^{t_{k+1}-t_{k+l_1-1}-1}A_{k+l_1-1}$,\\
$x^{t_{k+2}-t_{k+l_1}}A_{k+l_1}$, $x^{t_{k+2}-t_{k+l_1}+1}A_{k+l_1}$, $\cdots$, $x^{t_{k+1}-t_{k+l_1}-1}A_{k+l_1}\}$
$\cup$\\
$\{x^{t_{k+1}-t_{k+2}}A_{k+2}$, $x^{t_{k+1}-t_{k+2}+1}A_{k+2}$, $\cdots$, $x^{t_1-t_{k+2}-1}A_{k+2}$,\\
$x^{t_{k+1}-t_{k+3}}A_{k+3}$, $x^{t_{k+1}-t_{k+3}+1}A_{k+3}$, $\cdots$, $x^{t_1-t_{k+3}-1}A_{k+3}$,\\
$\vdots$\\
$x^{t_{k+1}-t_{k+l_1-i}}A_{k+l_1-i}$, $x^{t_{k+1}-t_{k+l_1-i}+1}A_{k+l_1-i}$, $\cdots$, $x^{t_1-t_{k+l_1-i}-1}A_{k+l_1-i}$,\\
$\vdots$\\
$x^{t_{k+1}-t_{k+l_1-1}}A_{k+l_1-1}$, $x^{t_{k+1}-t_{k+l_1-1}+1}A_{k+l_1-1}$, $\cdots$, $x^{t_1-t_{k+l_1-1}-1}A_{k+l_1-1}$,\\
$x^{t_{k+1}-t_{k+l_1}}A_{k+l_1}$, $x^{t_{k+1}-t_{k+l_1}+1}A_{k+l_1}$, $\cdots$, $x^{t_1-t_{k+l_1}-1}A_{k+l_1}\}$\\
$\cup$\\
$\{x^{t_1-t_{k+1}}A_{k+1}$, $x^{t_1-t_{k+1}+1}A_{k+1}$, $\cdots$, $x^{n-t_{k+1}-1}A_{k+1}$,\\
$x^{t_1-t_{k+2}}A_{k+2}$, $x^{t_1-t_{k+2}+1}A_{k+2}$, $\cdots$, $x^{n-t_{k+2}-1}A_{k+2}$,\\
$\vdots$\\
$x^{t_1-t_{k+l_1-i}}A_{k+l_1-i}$, $x^{t_1-t_{k+l_1-i}+1}A_{k+l_1-i}$, $\cdots$, $x^{n-t_{k+l_1-i}-1}A_{k+l_1-i}$,\\
$\vdots$\\
$x^{t_1-t_{k+l_1-1}}A_{k+l_1-1}$, $x^{t_1-t_{k+l_1-1}+1}A_{k+l_1-1}$, $\cdots$, $x^{n-t_{k+l_1-1}-1}A_{k+l_1-1}$,\\
$x^{t_1-t_{k+l_1}}A_{k+l_1}$, $x^{t_1-t_{k+l_1}+1}A_{k+l_1}$, $\cdots$, $x^{n-t_{k+l_1}-1}A_{k+l_1}\}$,\\
where $2 \leq j\leq r$ and $j$ takes the value in decreasing order from $r$ to $2$. In this case, to show that $B$ spans $B_2-B$ we use the same techniques as used in Case \ref{case-1} and Case \ref{case-2}. In the following table we summarize.
\begin{center}
\begin{tabular}{| l | l| c |}
\hline
Terms belongs to $B_2-B$ & Lemma used & Technique\\
&&similar to case\\
\hline
$x^{t_{2k-1}-t_{2k}}A_{2k}$ to & Lemma \ref{lm-t_k+i-1-t_k+i} for  & Case \ref{case-1}\\
$x^{t_{l_r+1}-t_{2k}-1}A_{2k}$ & $i=k,k-1, \cdots, l_r+2$ &\\
&&\\
$x^{t_{l_r+1}-t_{k+l_r+1}}A_{k+l_r+1}$ to & Lemma \ref{lm-t_i-t_k+i} for $i=l_r+1$  & Case \ref{case-2}\\
$x^{n-t_{k+l_r+1}-1}A_{k+l_r+1}$ &  &\\
&&\\
$x^{t_{l_r+1}-t_{k+l_r+2}}A_{k+l_r+2}$ to & Lemma \ref{lm-t_k+i-1-t_k+i} for  & Case \ref{case-1}\\
$x^{n-t_{2k}-1}A_{2k}$ & $i=k,k-1, \cdots, l_r+2$ &\\
\hline
for a fix $j$ where&&\\
$j=r, r-1, \cdots, 2$&&\\
&&\\
$x^{t_{k+l_j-1}-t_{k+l_j}}A_{k+l_j}$ to & Lemma \ref{lm-t_k+i-1-t_k+i} for $i=l_j,$ & Case \ref{case-1}\\
$x^{t_{l_{j-1}+1}-t_{k+l_j}-1}A_{k+l_j}$ & $l_j+1, \cdots, l_{j-1}+3,l_{j-1}+2$ &\\
&&\\
$x^{t_{l_{j-1}+1}-t_{k+l_{j-1}+1}}A_{k+l_{j-1}+1}$ to & Lemma \ref{lm-t_i-t_k+i} for $i=l_{j-1}+1$ & Case \ref{case-2}\\
$x^{n-t_{k+l_{j-1}+1}-1}A_{k+l_{j-1}+1}$&&\\
&&\\
$x^{t_{l_{j-1}+1}-t_{k+l_{j-1}+2}}A_{k+l_{j-1}+2}$ to & Lemma \ref{lm-t_k+i-1-t_k+i} for $i=l_j,$ & Case \ref{case-1}\\
$x^{n-t_{k+l_j}-1}A_{k+l_j}$ & $l_j+1, \cdots, l_{j-1}+3,l_{j-1}+2$&\\
\hline
$x^{t_{k+l_1-1}-t_{k+l_1}}A_{k+l_1}$ to & Lemma \ref{lm-t_k+i-1-t_k+i} for & Case \ref{case-1}\\
$x^{t_1-t_{k+l_1}-1}A_{k+l_1}$ & $i=l_1,l_1+1, \cdots, 3, 2$ &\\
&&\\
$x^{t_1-t_{k+1}}A_{k+1}$ to & Lemma \ref{lm-t_i-t_k+i} for $i=1$ & Case \ref{case-2}\\
$x^{n-t_{k+1}-1}A_{k+1}$ &  &\\
&&\\
$x^{t_1-t_{k+2}}A_{k+2}$ to & Lemma \ref{lm-t_k+i-1-t_k+i} for & Case \ref{case-1}\\
$x^{n-t_{k+l_1}-1}A_{k+l_1}$ & $i=l_1,l_1+1, \cdots, 3, 2$ &\\
\hline
\end{tabular}
\end{center}
\end{case}
This shows that $B$ spans $B_2$. Now we show that $B$ spans $B_1$. Recall that we have a homomorphism $\phi : C \rightarrow R_{u^k,p,n}$ (see Equation \eqref{surj-hom}). Therefore, $C/\text{Ker}\phi \simeq \phi(C)$ and $\phi(C)$ is a cyclic code over $R_{u^k,p}$. Thus, we can realize $C/\text{Ker}\phi$ as a cyclic code over $R_{u^k,p}$. Therefore, from Theorem 4.2 of \cite{Aks-Pkk13}, the minimal spanning set $B_{\phi(C)}$ of the code $C/\text{Ker}\phi$ is $\{A_1+\text{Ker}\phi, xA_1+\text{Ker}\phi$, $\cdots$, $x^{n-t_1-1}A_1+\text{Ker}\phi$, $A_2+\text{Ker}\phi, xA_2+\text{Ker}\phi$, $\cdots$, $x^{t_1-t_2-1}A_2+\text{Ker}\phi, A_3+\text{Ker}\phi, xA_3+\text{Ker}\phi$, $\cdots$, $x^{t_2-t_3-1}A_3+\text{Ker}\phi$, $\cdots$, $A_k+\text{Ker}\phi, xA_k+\text{Ker}\phi$, $\cdots$, $x^{t_{k-1}-t_k-1}A_k+\text{Ker}\phi$\}. To show $B$ spans $B_1$, we only show that $x^{t_1-t_2}A_2 \in \text{Span}(B)$. In a similar way, we can show that $x^{t_1-t_2+1}A_2$, $\cdots$, $x^{n-t_2-1}A_2,$ $\cdots$,  $x^{t_{k-1}-t_k}A_k, \
\cdots, x^{n-t_k-1}A_k \in \text{Span}(B).$ Since $B_{\phi(C)}$ spans $C/\text{Ker}\phi$, we can write $x^{t_1-t_2}A_2+\text{Ker}\phi$ as a $R_{u^k,p}$ linear combination of the elements of $B_{\phi(C)}$, i.e., $x^{t_1-t_2}A_2+\text{Ker}\phi$ = $\sum\limits_{i=0}^{n-t_1-1}c_{i1}(x^iA_1+\text{Ker}\phi)$ + $\cdots$ + $\sum\limits_{i=0}^{t_{k-1}-t_k-1}c_{ik}(x^iA_k+\text{Ker}\phi$), where $c_{ij} \in R_{u^k,p}$. Thus,
\[x^{t_1-t_2}A_2-\left(\sum\limits_{i=0}^{n-t_1-1}c_{i1}(x^iA_1) + \cdots + \sum\limits_{i=0}^{t_{k-1}-t_k-1}c_{ik}(x^iA_k)\right) \in \text{Ker}\phi.\]
Since $\text{Ker}\phi = \text{Span}(B_2)$ and $B$ spans $B_2$, we get $x^{t_1-t_2}A_2 \in \text{Span}(B)$. Similarly, we can show that $x^{t_1-t_2+1}A_2$, $\cdots$, $x^{n-t_2-1}A_2,$ $\cdots$,  $x^{t_{k-1}-t_k}A_k$, $\cdots$, $x^{n-t_k-1}A_k \in \text{Span}(B).$ This shows that $B$ spans $B_1$.

It is easy to see that any elements of the spanning set $B$ can not be written as the linear combination of its preceding elements and other elements in the spanning set $B$. Here we only show that $x^{t_1-t_2-1}A_2$ can not be written as linear combinations of others element of spanning set $B$. The proof is similar for the rest. Suppose, if possible $x^{t_1-t_2-1}A_2$ can be written as linear combinations of the others element of the spanning set $B$. Then we have 
$x^{t_1-t_2-1}A_2=\sum\limits_{i=0}^{n-t_1-1}\alpha_{1i}x^iA_1+\sum\limits_{i=0}^{t_1-t_2-2}\alpha_{2i}x^iA_2+\sum\limits_{i=0}^{t_2-t_3-1}\alpha_{3i}x^iA_3+ \cdots +\sum\limits_{i=0}^{t_{k-1}-t_k-1}\alpha_{ki}x^iA_k+\sum\limits_{i=0}^{t_1-t_{k+1}-1}\alpha_{(k+1)i}x^iA_{k+1}+\sum\limits_{i=0}^{t'_1-t_{k+2}-1}\alpha_{(k+2)i}x^iA_{k+2}+ \cdots +\sum\limits_{i=0}^{t'_{k-1}-t_{2k}-1}\alpha_{(2k)i}x^iA_{2k}$, where, $\alpha_{ij}=\sum\limits_{l=0}^{k-1}\beta_{ij}^{(l)}u^l+v\sum\limits_{m=0}^{k-1}\beta_{ij}^{(m)}u^m$, where, $\beta_{ij}^{(l)}, \beta_{ij}^{(m)} \in \F_p$ (Note that $l, m$ is not a power of $\beta$ it is a notation but $l, m$ is a power of $u$). Thus, $x^{t_1-t_2-1}(ug_2(x)+u^2g_{22}(x)+\cdots+u^{k-1}g_{2(k-1)}(x)+v(g_{2k}(x)+ug_{2(k+1)}(x)+\cdots+u^{k-1}g_{2(2k-1)}(x)))=g_1(x)\sum\limits_{j=0}^{n-t_1-1}\beta_{1j}^{(0)}x^j+ug_1(x)\sum\limits_{j=0}^{n-t_1-1}\beta_{1j}^{(1)}x^j+ug_{11}(x)\linebreak\sum\limits_{j=0}^{n-t_1-1}\beta_{1j}^{(0)}x^j+ug_2(x)\sum\limits_{j=0}^{t_1-t_2-2}\beta_{2j}^{(0)}
x^j+u^2m_1(x)
+u^3m_2(x)+\cdots+u^{k-1}m_{k-2}(x)+v(m_{k-1}(x)+um_k(x)+u^2m_{k+1}(x)+ \cdots +u^{k-1}m_{2k-2}(x))$, where, $m_1(x)$, $m_2(x)$, $\cdots$, $m_{2k-2}(x)$ is a polynomials in $\F_p[x]$. By comparing both sides, we have $\beta_{1j}^{(0)}=0$ for $0 \leq j \leq n-t_1-1$ and $x^{t_1-t_2-1}g_2(x)=g_1(x)\sum\limits_{j=0}^{n-t_1-1}\beta_{1j}^{(1)}x^j+g_2(x)\sum\limits_{j=0}^{t_1-t_2-2}\beta_{2j}^{(0)}x^j$. Note that $\text{deg}(x^{t_1-t_2-1}g_2(x))=t_1-1$ but $\text{deg}(g_1(x)\sum\limits_{j=0}^{n-t_1-1}\linebreak\beta_{1j}^{(1)}x^j) \geq t_1$ and $\text{deg}(g_2(x)\sum\limits_{j=0}^{t_1-t_2-2}\beta_{2j}^{(0)}x^j) \leq t_1-2$. Hence, this gives a contradiction.
\end{proof}

\begin{theo}
Let $n$ be a positive integer relatively prime to $p$ and $C$ be a cyclic code of length $n$ over the ring $R_{u^k,v^2,p}$. If  $C=\langle g_1(x)+ug_2(x)+\cdots+u^{k-1}g_k(x), v(g_{k+1}(x)+ug_{k+2}(x)+\cdots+u^{k-1}g_{2k}(x))\rangle$  with $t_1={\rm deg}(g_1(x))$, $t_{k+1}={\rm deg}(g_{k+1}(x))$, then $C$ has rank $n-t_{k+1}$. The minimal spanning set of $C$ is  $B=\{g_1(x)+ug_2(x)+\cdots+u^{k-1}g_k(x)$, $x(g_1(x)+ug_2(x)+\cdots+u^{k-1}g_k(x)),$ $\cdots$, $x^{n-t_1-1}(g_1(x)+ug_2(x)+\cdots+u^{k-1}g_k(x))$, $v(g_{k+1}(x)+ug_{k+2}(x)+\cdots+u^{k-1}g_{2k}(x))$, $x(v(g_{k+1}(x)+ug_{k+2}(x)+\cdots+u^{k-1}g_{2k}(x)))$, $\cdots$, $x^{t_1-t_{k+1}-1}\linebreak (v(g_{k+1}(x)+ug_{k+2}(x)+\cdots+u^{k-1}g_{2k}(x)))\}$.
\end{theo}
\begin{proof}
The proof is as similar as Theorem \ref{rank-main}.
\end{proof}

\section{Minimum distance} \label{md}
Let $n$ be a positive integer not relatively prime to $p$. Let $ C $ be a cyclic code of length $n$ over $R_{u^k,v^2,p}$. We have $C_{2k}=\{f(x)\in \F_p[x]~|~u^{k-1}vf(x)\in C\}=\langle g_{2k}(x)\rangle$ (See Page \pageref{C_i's}). Also, we know that $C_{2k}$ is a cyclic code over $\F_p$.

\begin{theorem} \label{md1}
Let $n$ be a positive integer not relatively prime to $p$. If $ C = \langle A_1, A_2, \cdots, A_{2k}\rangle$ is a cyclic code of length $n$ over the ring $R_{u^k,v^2,p}$ Then $w_{H}(C) = w_{H}(C_{2k})$.
\end{theorem}
\begin{proof}
Let $M(x,u,v)=m_0(x)+um_1(x)+ \cdots +u^{k-1}m_{k-1}+v(m_k(x)+um_{k+1}(x)+ \cdots +u^{k-1}m_{2k-1}(x)) \in C,$ where $m_0(x), m_1(x), \cdots, m_{2k-1}(x) \in \F_p[x]$. We have $u^{k-1}vM(x)=u^{k-1}vm_0(x)$, $w_{H}(u^{k-1}vM(x)) \leq w_{H}(M(x))$ and $u^{k-1}vC$ is subcode of $C$ with $w_{H}(u^{k-1}vC) \leq w_{H}(C)$. Thus  $w_{H}(u^{k-1}vC)=w_{H}(C)$. Therefore, it is sufficient to focus on the subcode $u^{k-1}vC$ in order to prove the theorem. Since $w_{H}(C_{2k})=w_{H}(u^{k-1}vC)$, we get $w_{H}(C)=w_{H}(C_{2k})$.
\end{proof}

\begin{definition}
Let $ m = b_{l-1}p^{l-1} + b_{l-2}p^{l-2} + \cdots + b_1p + b_0$, $b_i \in \F_p, 0 
\leq i \leq l-1$, be the $p$-adic expansion of $m$.
\begin{enumerate} [{\rm (1)}]
 \item If $ b_{l-i}  \neq 0$ for all $1  \leq i \leq q, q < l, $ and $ b_{l-i} = 0 $ for all $i, q+1 \leq i \leq l$, then $m$ is said to have a $p$-adic length $q$ zero expansion.
\item If $ b_{l-i}  \neq 0$ for all $1  \leq i \leq q, q < l, $ $b_{l-q-1} = 0$ and $ b_{l-i} \neq 0 $ for some $i, q+2 \leq i \leq l$, then $m$ is said to have  $p$-adic length $q$ non-zero expansion.
\item If $ b_{l-i}  \neq 0$ for $1  \leq i \leq l, $ then $m$ is said to have a $p$-adic length $l$  expansion or $p$-adic full expansion.
\end{enumerate}
\end{definition}

\begin{lemma} \label{lm-md}
Let $C$ be a cyclic code over $R_{u^k,v^2,p}$ of length $p^l$ where $l$ is a positive integer. Let $C = \langle g(x)\rangle$ where $g(x) = (x^{p^{l-1}}-1)^bh(x)$, $ 1 \leq b < p$. If $h(x)$ generates a cyclic code of length $p^{l-1}$ and minimum distance $d$ then the minimum distance $d(C)$ of $C$ is $(b+1)d$.
\end{lemma}
\begin{proof}
For $c \in C$, we have $c=(x^{p^{l-1}}-1)^bh(x)m(x)$ for some $ m(x) \in \frac{R_{u^k,v^2,p}[x]}{\langle x^{p^l}-1\rangle}$. Since $h(x)$ generates a cyclic code of length $p^{l-1}$, we have $w(c) = w((x^{p^{l-1}} - 1)^bh(x)m(x)) = w(x^{p^{l-1}b}h(x)m(x))+w(^bC_1x^{p^{l-1}(b-1)}h(x)m(x)) + \cdots + w(^bC_{b-1}\\x^{p^{l-1}}h(x)m(x)) + w(h(x)m(x))$. Thus, $ d(C) = (b + 1)d$.
\end{proof}

\begin{theorem} \label{md-thm}
Let $C$ be a cyclic code over the ring $R_{u^k,v^2,p}$ of length $p^l$ where $l$ is a positive integer. Then,  $C = \langle A_1, A_2, \cdots, A_{2k}\rangle$ where $g_1(x) = (x-1)^{t_1}, g_2(x) = (x-1)^{t_2}, \cdots,  g_{2k}(x) = (x-1)^{t_{2k}}$ $(A_i$'s and $g_i(x)$'s are defined in page \pageref{A_i's} $($see page \pageref{A_i's}$))$ for some $t_1 > t_2 > \cdots > t_k> 0$, $t_{k+1} > t_{k+2} > \cdots > t_{2k} > 0$ and $t_i > t_{k+i}$ for $1 \leq i \leq k$ 
\begin{enumerate}[{\rm (1)}]
\item If $t_{2k} \leq p^{l-1},$ then $d(C) = 2$. 
\item If $t_{2k} > p^{l-1}$, let $t_{2k} = b_{l-1}p^{l-1} + b_{l-2}p^{l-2} + \cdots + b_1p + b_0$ be the $p$-adic expansion of $t_{2k}$ and $ g_{2k}(x) = (x-1)^{t_{2k}} = (x^{p^{l-1}} - 1)^{b_{l-1}}(x^{p^{l-2}} - 1)^{b_{l-2}} \cdots (x^{p^{1}} - 1)^{b_1}(x^{p^0} - 1)^{b_0}$.
\begin{enumerate}[{\rm ($a$)}]
 \item If $t_{2k}$ has a $p$-adic length $q$ zero expansion or full expansion $(l=q)$, then $d(C) = (b_{l-1}+1)(b_{l-2}+1)\cdots(b_{l-q}+1)$.
\item If $t_{2k}$ has a $p$-adic length $q$ non-zero expansion, then $d(C) = 2(b_{l-1}+1)(b_{l-2}+1)\cdots(b_{l-q}+1)$.
\end{enumerate}
\end{enumerate}
\end{theorem}
\begin{proof}
The first claim easily follows from Theorem \ref{properties}. From Theorem \ref{md1}, we see that $d(C)=d(C_{2k})=d(\langle(x-1)^{t_{2k}}\rangle)$. Hence, we only need to determine the minimum weight of $C_{2k}= \langle(x-1)^{t_{2k}}\rangle$.\\
(1) If $t_{2k} \leq p^{l-1},$ then $(x-1)^{t_{2k}}(x-1)^{p^{l-1}-t_{2k}} = (x-1)^{p^{l-1}}=(x^{p^{l-1}}-1) \in C$. Thus, $d(C)=2$.\\
(2) Let $ t_{2k}>p^{l-1}$. (a) If $t_{2k}$ has a $p$-adic length $q$ zero expansion, we have $t_{2k}=b_{l-1}p^{l-1}+b_{l-2}p^{l-2} + \cdots + b_{l-q}p^{l-q}$, and $g_{2k}(x)=(x - 1)^{t_{2k}}=(x^{p^{l-1}}-1)^{b_{l-1}}(x^{p^{l-2}}-1)^{b_{l-2}}\cdots(x^{p^{l-q}}-1)^{b_{l-q}}$. Let $h(x)=(x^{p^{l-q}}-1)^{b_{l-q}}$. Then $h(x)$ generates a cyclic code of length $p^{l-q+1}$ and minimum distance $(b_{l-q}+1)$. By Lemma \ref{lm-md}, the subcode generated by $(x^{p^{l-q+1}}-1)^{b_{l-q+1}}h(x)$ has minimum distance $(b_{l-q+1}+1)(b_{l-q}+1)$. By induction on $q$, we can see that the code generated by $g_{2k}(x)$ has minimum distance $(b_{l-1}+1(b_{l-2}+1)\cdots(b_{l-q}+1)$. Thus, $d(C)=(b_{l-1}+1)(b_{l-2}+1)\cdots(b_{l-q}+1)$.\\
(b) If $t_{2k}$ has a $p$-adic length $q$ non-zero expansion, we have $t_{2k}=b_{l-1} p^{l-1} + b_{l-2}p^{l-2} + \cdots + b_{1}p + b_0, b_{l-q-1}=0$. Let $r=b_{l-q-2}p^{l-q-2}+b_{l-q-3}p^{l-q-3}+ \cdots + b_1p + b_0$ and $h(x)=(x-1)^r=(x^{p^{l-q-2}}-1)^{b_{l-q-2}}(x^{p^{l-q-3}}-1)^{b_{l-q-3}}\cdots(x^{p^{1}}-1)^{b_{1}}(x^{p^{0}}-1)^{b_{0}}$. Since $r < p^{l-q-1}$, we have $p^{l-q-1}=r+j$ for some non-zero $j$. Thus, $(x-1)^{p^{l-q-1}-j}h(x)=(x^{p^{l-q-1}}-1) \in C$. Hence, the subcode generated by $h(x)$ has minimum distance 2. By Lemma \ref{lm-md}, the subcode generated by $(x^{p^{l-q}}-1)^{b_{l-q}}h(x)$ has minimum distance $2(b_{l-q}+1)$. By induction on $q$, we can see that the code generated by $g_{2k}(x)$ has minimum distance $2(b_{l-1}+1)(b_{l-2}+1)\cdots(b_{l-q}+1)$. Thus, $d(C) = 2(b_{l-1}+1)(b_{l-2}+1)\cdots(b_{l-q}+1)$.\\
\end{proof}

\section{Examples} \label{exm}
\begin{example}
Cyclic codes of length $4$ over the ring $R_{u^3,v^2,2}=\F_2+u\F_2+u^2\F_2+v(\F_2+u\F_2+u^2\F_2)$, $u^3=0$, $v^2=0$, $uv=vu$: We have
\begin{equation}
x^4-1=(x-1)^4 ~ \text{over} ~ \F_2\notag
\end{equation}
Let $g=x-1$ and $c_0,c_1,\cdots, c_{11}\in \F_2$. The some of the non zero cyclic codes of length 4 over the ring $R_{u^3,v^2,2}$ with generator polynomials, rank and minimum distance are given in Tables 1 and 2.
\end{example}

\newpage
\begin{center}
{\bf Table 1.} Some non zero cyclic codes of length 4 over $R_{u^3,v^2,2}$.\\
\begin{tabular}{| l | c| c |}
\hline
Non-zero generator polynomials & Rank & d(C)\\
\hline
$\langle vu^2g^3\rangle$ & 1 & 4\\
\hline
$\langle v(ug^3+u^2c_0g), vu^2g^2\rangle$ & 2 & 2\\
\hline
$\langle v(g^3+uc_0g+u^2c_1), v(ug^2+u^2c_2), vu^2g\rangle$ & 3 & 2\\
\hline
$\langle u^2g^3+v(c_0g^2+uc_1g+u^2c_2), v(g^3+uc_3g+u^2c_4),$ & 4 & 2\\
$v(ug^2+u^2c_5), vu^2g\rangle$&  & \\
\hline
$\langle ug^3+u^2c_0g+v(c_1g^2+uc_2g+u^2c_3),$&  & \\
$ u^2g^2+v(c_4g^2+uc_5g+u^2c_6), v(g^3+uc_7g+u^2c_8),$ & 5 & 2\\
$v(ug^2+u^2c_9), vu^2g\rangle$&  & \\
\hline
$\langle g^3+uc_0g+u^2c_1+v(c_2g^2+uc_3g+u^2c_4),$&  & \\
$ug^2+u^2c_5+v(c_6g^2+uc_7g+u^2c_8),$ & 3 & 2\\
$u^2g+v(c_9g^2+uc_{10}g+u^2c_{11}),\rangle$&  & \\
\hline
\end{tabular}
\end{center}

\begin{center}
{\bf Table 2.} Non zero free cyclic codes of length 4 over $R_{u^3,v^2,2}$.\\
\begin{tabular}{| l | c| c |}
\hline
Non-zero generator polynomials & Rank & d(C)\\
\hline
$\langle g^3+uc_0g^2+u^2c_1g^2+v(c_2g^2+uc_3g^2+u^2c_4g^2)\rangle$ & 1 & 4\\
\hline
$\langle g^2+u(c_0+c_1x)+u^2(c_2+c_3x)g^2$ & 2 & 2\\
$+v((c_4+c_5x)g^2+u(c_6+c_7x)g^2+u^2(c_8+c_9x)g^2)\rangle$ &  & \\
\hline
$\langle g+uc_0+u^2c_1+v(c_2+uc_3+u^2c_4)\rangle$ & 3 & 2\\
\hline
$\langle 1\rangle$ & 4 & 1\\
\hline
\end{tabular}
\end{center}

\begin{example}
Cyclic codes of length $4$ over the ring $R_{u^3,v^2,3}=\F_3+u\F_3+u^2\F_3+v(\F_3+u\F_3+u^2\F_3)$, $u^3=0$, $v^2=0$, $uv=vu$: We have
\begin{equation}
x^4-1=(x+1)(x+2)(x^2+1) ~ \text{over} ~ \F_3\notag
\end{equation}
Let $g_1=x+1$, $g_2=x+2$ and $g_3=x^2+1$. The non zero cyclic codes of length 4 over the ring $R_{u^3,v^2,3}$ with generator polynomials and rank  given in Table 3.
\end{example}

\begin{center}
{\bf Table 3.} Non zero cyclic codes of length 4 over $R_{u^3,v^2,3}$.\\
\begin{tabular}{| l | c|}
\hline
Non-zero generator polynomials & Rank\\
\hline
$\langle g_1g_2+ug_1g_2+u^2g_1, v(g_1g_2+ug_2+u^2)\rangle$ & 2 \\
\hline
$\langle g_1g_2+ug_1+u^2, v(g_2+u+u^2)\rangle$ & 3 \\
\hline
$\langle g_1g_3+ug_3+u^2, v(g_1+u+u^2)\rangle$ & 3 \\
\hline
$\langle ug_2g_3+u^2g_3, v(g_2g_3+ug_3+u^2g_3)\rangle$ & 1  \\
\hline
$\langle v(g_1g_3+ug_3+u^2)\rangle$ & 1 \\
\hline
$\langle g_2g_3+ug_2+u^2\rangle$ & 1 \\
\hline

\end{tabular}
\end{center}

%In Table 2, we give examples of optimal ternary codes obtained as the Gray images of cyclic codes over $R_{u^2,v^2,3}$. In Table 2, $[~.~]^*$ will denote that the ternary code is an optimal code. From Table 2, we can see that we have obtained all ternary optimal code except $[12, 2, 9]^*$.\\

%{\bf Table 2.}  Ternary images of some cyclic codes of length 3 over $R_{u^2,v^2,3}$.\\
%\begin{center}
%\begin{tabular}{| l | c| c |}
%\hline
%Non-zero generator polynomials & $\phi_L(C)$\\
%\hline
%$<uvg^2>$ & $[12, 1, 12]^*$\\
%\hline

%\end{tabular}
%\end{center}

\bibliographystyle{plain}
\bibliography{ref}
\end{document}